\documentclass[preprint,12pt]{elsarticle}

\usepackage[margin=1in]{geometry}

\usepackage{graphicx}

\usepackage{amssymb}

\usepackage{url}
\usepackage{hyperref}

\usepackage{amsmath}
\usepackage{upgreek}

\usepackage{amsthm}

\def\doubleunderline#1{\underline{\underline{#1}}}
\usepackage{tensor}
\usepackage{dsfont}

\usepackage{siunitx}

\usepackage{amssymb}

\usepackage{bbm}

\usepackage{soul}

\usepackage{algorithmicx}
\usepackage{algorithm}
\usepackage[noend]{algpseudocode}
\makeatletter
\renewcommand{\ALG@beginalgorithmic}{\footnotesize}
\makeatother

\usepackage[dvipsnames,table,xcdraw]{xcolor}

\usepackage{booktabs}

\usepackage{amsmath}

\makeatletter
\newcommand{\vast}{\bBigg@{4}}
\newcommand{\Vast}{\bBigg@{7}}
\makeatother

\usepackage{caption}

\usepackage{subcaption}

\usepackage{scalerel,stackengine}
\stackMath
\newcommand\reallywidehat[1]{%
\savestack{\tmpbox}{\stretchto{%
  \scaleto{%
    \scalerel*[\widthof{\ensuremath{#1}}]{\kern-.6pt\bigwedge\kern-.6pt}%
    {\rule[-\textheight/2]{1ex}{\textheight}}
  }{\textheight}%
}{0.5ex}}%
\stackon[1pt]{#1}{\tmpbox}%
}

\newcommand*\xbar[1]{%
   \hbox{%
     \vbox{%
       \hrule height 0.5pt 
       \kern0.5ex
       \hbox{%
         \kern-0.1em
         \ensuremath{#1}%
         \kern-0.1em
       }%
     }%
   }%
} 

\date{}

\usepackage{todonotes}

\setcitestyle{authoryear}
 \biboptions{comma,round}

\journal{Journal of Computational Physics}

\begin{document}

\begin{frontmatter}

\title{A kinetic energy--and entropy-preserving scheme for\\ compressible two-phase flows}



\author{Suhas S. Jain\corref{cor1}}
\ead{sjsuresh@stanford.edu}
\cortext[cor1]{Corresponding author}
\author{Parviz Moin}
\ead{moin@stanford.edu}
\address{Center for Turbulence Research, Stanford University, California, USA 94305}

\begin{abstract}

Accurate numerical modeling of compressible flows, particularly in the turbulent regime, requires a method that is non-dissipative and stable at high Reynolds ($Re$) numbers. For a compressible flow, it is known that discrete conservation of kinetic energy is not a sufficient condition for numerical stability, unlike in incompressible flows. 

In this study, we adopt the recently developed conservative diffuse-interface method (Jain, Mani \& Moin, \textit{J. Comput. Phys.}, 2020) along with the five-equation model for the simulation of compressible two-phase flows. This method discretely conserves the mass of each phase, momentum, and total energy of the system and consistently reduces to the single-phase Navier-Stokes system when the properties of the two phases are identical. We here propose discrete consistency conditions between the numerical fluxes, such that any set of numerical fluxes that satisfy these conditions would not spuriously contribute to the kinetic energy and entropy of the system. We also present a set of numerical fluxes\textemdash which satisfies these consistency conditions\textemdash that results in an exact conservation of kinetic energy and approximate conservation of entropy in the absence of pressure work, viscosity, thermal diffusion effects, and time-discretization errors. Since the model consistently reduces to the single-phase Navier-Stokes system when the properties of the two phases are identical, the proposed consistency conditions and numerical fluxes are also applicable for a broader class of single-phase flows.

To this end, we present coarse-grid numerical simulations of compressible single-phase and two-phase turbulent flows at infinite $Re$, to illustrate the stability of the proposed method in canonical test cases, such as an isotropic turbulence and Taylor-Green vortex flows. A higher-resolution simulation of a droplet-laden compressible decaying isotropic turbulence is also presented, and the effect of the presence of droplets on the flow is analyzed.

\end{abstract}


\end{frontmatter}



\section{Introduction}



Compressible turbulent flows are encountered in a wide range of aerodynamic flows and other high-speed flow applications.
Particularly for two-phase flows, the applications include atomization \citep{sallam2006primary}, droplet combustion \citep{law1982recent}, bubble cavitation \citep{plesset1977bubble}, Rayleigh-Taylor instability \citep{youngs1984numerical}, and Richtmyer-Meshkov instability \citep{brouillette2002richtmyer} flows. 
Accurate numerical modeling of such flows requires a method that is non-dissipative and is stable at high Reynolds numbers ($Re$). For simulations of turbulent flows, it is known that the correct representation of evolution of kinetic energy is important for achieving accurate numerical simulations. For incompressible simulations, \citet{morinishi1998fully} showed that the numerical instabilities can be suppressed in the limit of zero viscosity without adding numerical dissipation if kinetic energy is discretely conserved. For a compressible flow, however, it is known that discrete conservation of kinetic energy is not a sufficient condition for numerical stability, unlike in incompressible flows, and additional constraints on the entropy of the system are needed to suppress the non-linear instabilities \citep{honein2005}.

The first class of methods that improve numerical stability consists of those that try to minimize aliasing errors. \citet{feiereisen1983numerical} proposed using skew-symmetric splitting for convective terms to ensure that the method does not spuriously contribute to the kinetic energy of the system. \citet{blaisdell1996effect}, \citet{kravchenko1997effect}, and \citet{ducros2000high} showed that the use of skew-symmetric splitting for the convective terms reduces the aliasing errors compared to the convective form of the discretization. \citet{lee1993large} showed that the use of the internal energy equation instead of the total energy equation results in better numerical stability, and \citet{nagarajan2003robust} used a staggered scheme to show better numerical stability compared to a collocated scheme. The skew-symmetric formulations have since been used widely to achieve improved numerical stability for the simulations of compressible turbulent flows \citep[see,][]{subbareddy2009fully,kok2009high,morinishi2010skew, pirozzoli2010generalized,pirozzoli2011stabilized}.
However, these methods are known to become unstable at {coarse resolutions for} high $Re$, even at low Mach numbers and in the absence of shocks, without a subgrid-scale model. {For high enough resolutions, in the limit of direct numerical simulation, all these methods can yield stable numerical simulations.}

Another class of methods consists of those that try to suppress numerical instabilities by satisfying an entropy condition. Schemes that satisfy entropy condition are known to exhibit stable density fluctuations. \citet{honein2004higher} proposed a modification to the total energy equation that mimics solving an entropy transport equation with skew-symmetric splitting, and showed that the modified total energy equation results in higher conservation of entropy and excellent numerical stability at high $Re$. \citet{tadmor1987numerical,tadmor2003entropy} proposed entropy-conservative numerical fluxes that can be combined with additional dissipation to obtain a method that satisfies a mathematical entropy condition for a general system of conservation laws. Following this idea, and choosing the physical thermodynamic entropy as the mathematical entropy function, \citet{chandrashekar2013kinetic} derived numerical fluxes that conserve kinetic energy and entropy for the Euler and Navier-Stokes equations. More recently, \citet{coppola2019numerically, kuya2018kinetic} proposed schemes that have good kinetic energy and entropy conservation properties. Hereafter, we refer to these class of schemes as kinetic energy--and entropy-preserving (KEEP) schemes. 


It is worth mentioning that in the context of incompressible two-phase flows, various methods have been proposed that conserve kinetic energy of the system \citep{fuster2013energy,valle2020energy,mirjalili2020consistent,jain2022accurate} and therefore are stable for high-$Re$ flows. For compressible two-phase flows, on the other hand, it was found that the use of existing KEEP schemes in the literature, that were developed for single-phase flows, does not result in stable numerical simulations \citep{jain2020keep}. In addition to the conservation of entropy and kinetic energy, the non-linear stability properties, such as boundedness and total-variation diminishing (TVD) properties of the volume-fraction field, are important to maintain realizable values of the other field variables, such as density of the two-phase mixture. We also need thermodynamic consistency at the interface to avoid spurious numerical solutions. It is even more crucial to have all these non-linear stability properties, if one is using a non-dissipative scheme, to obtain stable numerical simulations. Therefore, unlike single-phase compressible flows, entropy stability alone does not imply numerical stability for two-phase compressible flows. 

The focus of the present work is, therefore, on the development of a robust, numerically stable scheme\textemdash a KEEP scheme\textemdash that works for both single-phase and two-phase compressible flows. We adopt the recently developed conservative diffuse-interface method \citep{jain2020conservative} along with the five-equation model \citep{allaire2002five,kapila2001two} for modeling the system of compressible two-phase flows. Some favorable properties of this method are: (a) it discretely conserves the mass of each phase, momentum, and total energy of the system, (b) it satisfies the boundedness and TVD properties for the volume-fraction field, (c) it has consistency corrections (consistent regularization) in the momentum and energy equations that maintain mechanical equilibrium at the interface, and (d) it consistently reduces to the single-phase Navier-Stokes system when the properties of the two phases are identical. The proposed framework of KEEP scheme presented in this work, however, can be adopted for any of the other interface-capturing methods, such as a volume-of-fluid or a level-set method. 
We, here, propose a general framework for the development of KEEP schemes for single-phase and two-phase flows. We propose discrete consistency conditions for the numerical fluxes of volume fraction, mass, momentum, kinetic energy, and internal energy, such that an exact conservation of kinetic energy and an approximate conservation of entropy are achieved in the absence of pressure work, viscosity, thermal diffusion effects, and time-differencing errors. To this end, we present numerical simulations of single-phase and two-phase turbulent flows at finite and infinite $Re$ to illustrate the stability and accuracy of the method in isotropic turbulence and Taylor-Green vortex flows. 

The rest of this paper is organized as follows. The conservative diffuse-interface method used in this work with a five-equation model is described in Section \ref{sec:model} along with the derivation of analytical transport equations for kinetic energy and entropy. The discrete consistency conditions between the numerical fluxes are presented in Section \ref{sec:keep_scheme} along with the proposed KEEP scheme. 
The under-resolved simulations of single-phase and two-phase Taylor-Green vortex and isotropic turbulence at infinite $Re$ are presented in Section \ref{sec:under_resolved};
and the resolved simulations of droplet-laden isotropic turbulence at finite $Re$ are presented in Section \ref{sec:resolved}, 
followed by a summary and concluding remarks in Section \ref{sec:conclusion}.

\section{Conservative diffuse-interface method \label{sec:model}}

The conservative diffuse-interface method for compressible two-phase flows \citep{jain2020conservative} with the underlying five-equation model is given in Eqs. \eqref{eq:volumef}-\eqref{eq:energyf} along with the mixture equation of state (EOS) in Eq. \eqref{eq:pressuref}. 
\begin{equation}
\frac{\partial \phi_1}{\partial t} + \frac{\partial u_j \phi_1}{\partial x_j} = (\phi_1 + \zeta_1)\frac{\partial u_j}{\partial x_j} + \frac{\partial a_{1j}}{\partial x_j}, 
\label{eq:volumef}    
\end{equation}
\begin{equation}
\frac{\partial m_l}{\partial t} + \frac{\partial u_j m_l}{\partial x_j} = \frac{\partial R_{lj}}{\partial x_j},  \hspace{0.5cm} l=1,2,
\label{eq:massf}
\end{equation}
\begin{equation}
\frac{\partial \rho u_i}{\partial t} + \frac{\partial \rho u_i u_j}{\partial x_j} + \frac{\partial p}{\partial x_i} = \frac{\partial u_i f_j}{\partial x_j} + \frac{\partial \tau_{ij}}{\partial x_j} + \sigma \kappa \frac{\partial \phi_1}{\partial x_i} + \rho g_i,
\label{eq:momf}
\end{equation}
\begin{equation}
\frac{\partial E}{\partial t} + \frac{\partial \left(E + p \right) u_j}{\partial x_j} = \frac{\partial \tau_{ij} u_i}{\partial x_j} + \frac{\partial k f_j}{\partial x_j} + \sum_{l=1}^2 \frac{\partial \rho_l h_l a_{lj}}{\partial x_j} + \sigma \kappa u_i\frac{\partial \phi_1}{\partial x_i} + \rho u_i g_i,
\label{eq:energyf}
\end{equation}
\begin{equation}
p = \frac{\rho e + \left(\sum_{l=1}^2\frac{\phi_l \beta_l}{\alpha_l} \right)} {\left(\sum_{l=1}^2\frac{\phi_l}{\alpha_l} \right)}.
\label{eq:pressuref}
\end{equation}
This form of the model has a volume-fraction advection equation [Eq. \eqref{eq:volumef}], a mass balance equation for each of the phases $l$ [Eq. \eqref{eq:massf}], a momentum equation [Eq. \eqref{eq:momf}], and a total energy equation [Eq. \eqref{eq:energyf}]. If a general EOS for phase $l$ is written as $p_l=\alpha_l\rho_le_l + \beta_l$, where $\alpha_l$ and $\beta_l$ are parameters in the EOS that are generally determined using experimental measurements, then by invoking the isobaric closure law for pressure in the mixture region 
\begin{equation}
  p = p_1 = p_2,
  \label{eq:isobar}
\end{equation}
the generalized mixture EOS can be written as in Eq. (\ref{eq:pressuref}). {Here, a linear stiffened-gas EOS has been assumed for simplicity. However, following the same procedure, a generalized mixture EOS can be written for any complex cubic EOS.}

We specifically use the five-equation model by \citet{kapila2001two} over the \citet{allaire2002five} model, which includes an additional dilatational term in the volume-fraction equation, due to its improved accuracy for flows that involve high compression and expansion, as illustrated by \citet{tiwari2013diffuse} and \citet{schmidmayer2020assessment}. \citet{murrone2005five} also showed that the five-equation model by \citet{kapila2001two} admits conservative transport equations for entropy of each phase. For weakly compressible flows, we found that both five-equation models yield identical results.


In Eqs. \eqref{eq:volumef}-\eqref{eq:pressuref}, $\phi_l$ is the volume fraction of phase $l$ that satisfies the condition $\sum_{l=1}^2 \phi_l=1$; $\rho_l$ is the density of phase $l$; $\rho$ is the total density, defined as $\rho=\sum_{l=1}^2\rho_l\phi_l$; $u_i$ is the velocity; $p$ is the pressure; $e$ is the specific mixture internal energy, which can be related to the specific internal energy of phase $l$, $e_l$, as $e=\sum_{l=1}^2 \rho_le_l$; $k=u_iu_i/2$ is the specific kinetic energy; $E=\rho(e+k)$ is the total energy of the mixture per unit volume; and the function $\zeta_1$ is given by
\begin{equation}
    \zeta_1 = \frac{\rho_2 c_2^2 - \rho_1 c_1^2}{\frac{\rho_1 c_1^2}{\phi_1} + \frac{\rho_2 c_2^2}{\phi_2}},
\end{equation}
where $c_l$ is the speed of sound for phase $l$. If each of the phases is assumed to follow a stiffened-gas EOS, then the parameters in the EOS can be written as $\alpha=\gamma - 1$ and $\beta = -\gamma \pi$, where $\gamma$ is the polytropic coefficient and $\pi$ is the reference pressure. Then, $c_l$ can be defined as
\begin{equation}
    c_l=\sqrt{\gamma_l\Big(\frac{p + \pi_l}{\rho_l}\Big)}.
\end{equation}
In Eq. (\ref{eq:energyf}), $h_l=e_l+p/\rho_l$ represents the specific enthalpy of the phase $l$ and can be expressed in terms of $\rho_l$ and $p$ using the stiffened-gas EOS as
\begin{equation}
    h_l=\frac{(p + \pi_l)\gamma_l}{\rho_l(\gamma_l - 1)}.
    \label{eq:enthalpy}
\end{equation}

In Eqs. (\ref{eq:volumef})-(\ref{eq:energyf}), $\sigma$ is the surface-tension coefficient, $\kappa=-\partial n_{1j}/\partial x_j$ is the curvature of the interface, $g_i$ is the gravitational acceleration, and $a_{1i}=a_i(\phi_1)=\Gamma\{\epsilon\partial\phi_1/\partial x_i - \phi_1(1 - \phi_1)n_{1i}\}$ is the volumetric interface-regularization flux for phase $1$, which satisfies the condition $a_i(\phi_1)=-a_i(\phi_2)$. $n_{1i}=(\partial\phi_1/\partial x_i)/|(\partial \phi_1/\partial x_i)|$ is the normal of the interface for phase $1$; and $\Gamma$ and $\epsilon$ are the interface parameters, where $\Gamma$ represents a regularization velocity scale and $\epsilon$ represents an interface thickness scale. $R_{li}=\rho_l a_{li}$ is the interface-regularization flux in the mass balance equation for phase $l$, and $m_l=\rho_l\phi_l$ is the mass per unit total volume for phase $l$. Summing up the mass balance equations for phases $1$ and $2$ in Eq. \eqref{eq:massf}, the total mixture mass balance equation (a modified continuity equation) can be written as 
\begin{equation}
\frac{\partial \rho}{\partial t} + \frac{\partial \rho u_j}{\partial x_j} = \frac{\partial f_j}{\partial x_j},
\label{eq:mod_continuity}
\end{equation}
where $f_i=\sum_{l=1}^2 R_{li}=\sum_{l=1}^2 \rho_l a_{li}$ is the net interface-regularization flux for the mixture mass equation. 

Invoking Stokes's hypothesis, we write the Cauchy stress tensor as $\tau_{ij} = 2\mu D_{ij} + (\beta - 2\mu/3)(\partial u_k/\partial x_k) \delta_{ij}$, where $\mu$ is the dynamic viscosity of the mixture evaluated using the one-fluid mixture rule as $\mu=\sum_{l=1}^2 \phi_l \mu_l$; $D_{ij} =\{\partial u_i/\partial x_j + \partial u_j/\partial x_i\}/2$ is the strain-rate tensor; and $\beta$ is the bulk viscosity.


\subsection{Transport equations for kinetic energy and internal energy \label{sec:ke_analytical}}

The transport equation for kinetic energy can be obtained by multiplying $u_i$ to the momentum equation. In the absence of viscosity, surface tension, and gravity effects, this can be written as


\begin{equation}
u_i\left(\frac{\partial \rho u_i}{\partial t} + \frac{\partial \rho u_i u_j}{\partial x_j} + \frac{\partial p}{\partial x_i} = \frac{\partial u_i f_j}{\partial x_j}\right).
\label{eq:u_dot_mom}
\end{equation}
Invoking the product rule and algebraically manipulating the equation, we arrive at the form
\begin{equation}
\begin{aligned}
\frac{\partial }{\partial t}\left(\rho \frac{u_iu_i}{2}\right) + \left(\frac{u_iu_i}{2}\right)\left\{\frac{\partial \rho}{\partial t} + \frac{\partial \rho u_j}{\partial x_j} - \frac{\partial f_j}{\partial x_j}\right\}\\ +
\frac{\partial}{\partial x_j} \left(\rho u_j \frac{u_i u_i}{2} \right)
- \frac{\partial}{\partial x_j} \left(f_j \frac{u_i u_i}{2} \right)
+ \frac{\partial p u_j}{\partial x_j} - p \frac{\partial u_j}{\partial x_j} = 0.
\end{aligned}
\label{eq:intermediate_kinetic}
\end{equation}
Now, subtracting out the continuity equation [Eq. \eqref{eq:mod_continuity}] multiplied by the term $u_iu_i/2$, we obtain the compressible two-phase kinetic energy transport equation:
\begin{equation}
\begin{aligned}
\frac{\partial }{\partial t}\left(\rho \frac{u_iu_i}{2}\right) + 
\frac{\partial}{\partial x_j} \left(\rho u_j \frac{u_i u_i}{2} \right)
- \frac{\partial}{\partial x_j} \left(f_j \frac{u_i u_i}{2} \right)
+ \frac{\partial p u_j}{\partial x_j} - p \frac{\partial u_j}{\partial x_j}
= 0.
\end{aligned}
\label{eq:kineticf}
\end{equation}
The kinetic energy is not a conserved quantity in the absence of viscosity, surface tension, and gravity force, as can be seen from the governing equation due to the pressure term, which is not in a divergence form. The pressure term represents the reversible exchange of energy between the kinetic and internal energy. However, the convective and interface-regularization terms are in a divergence form and their contribution to kinetic energy should be either only through the boundary terms or zero, in a periodic domain. Maintaining this property discretely is important for these terms to not contribute spuriously to the kinetic energy. 

The internal energy transport equation was derived by \citet{jain2020conservative}, and can be written as
\begin{equation}
    \frac{\partial \rho e}{\partial t} + \frac{\partial \rho e u_j}{\partial x_j} + 
    \frac{\partial p u_j}{\partial x_j} - u_j \frac{\partial p}{\partial x_j} = \sum_{l=1}^2 \frac{\partial \rho_l h_l a_{lj}}{\partial x_j},
    \label{eq:ief}
\end{equation}
in the inviscid limit. Summing up the internal energy and the kinetic energy transport equations, we can arrive at the total energy equation in Eq. \eqref{eq:energyf}.

\subsection{Transport equations for individual phase and mixture entropy \label{sec:entropy_analytical}}




Starting from the internal energy equation in Eq. \eqref{eq:ief} and rewriting this in terms of material derivative, we obtain
\begin{equation*}
    \frac{D \rho e}{D t} + \rho h\frac{\partial u_j}{\partial x_j} = \sum_{l=1}^2 \frac{\partial \rho_l h_l a_{lj}}{\partial x_j},
    \label{eq:conservative_ie}
\end{equation*}
where $D/Dt=\partial/\partial t + u_j\partial/\partial x_j$ represents a material derivative. Now, expressing the mixture internal energy in terms of phase quantities as
\begin{equation*}
    \frac{D \rho e}{D t} = \sum_{l=1}^2\frac{D (\phi_l\rho_l e_l)}{D t} = \sum_{l=1}^2 \left\{\phi_l\frac{D (\rho_l e_l)}{D t} + \rho_le_l\frac{D \phi_l}{D t}\right\}, 
\end{equation*}
and then utilizing Gibbs's relation to express internal energy in terms of specific entropy of each phase, $s_l$, as
\begin{equation*}
    \mathrm{d}(\rho_le_l) = \rho_l\mathrm{d}e_l + e_l\mathrm{d}\rho_l = \rho_lT_l\mathrm{d}s_l + h_l\mathrm{d}\rho_l,
\end{equation*}
we arrive at the relation
\begin{equation}
    \sum_{l=1}^2 \left\{\rho_l\phi_lT_l \frac{Ds_l}{Dt} + \phi_lh_l\frac{D\rho_l}{Dt} + \phi_lh_l\rho_l\frac{\partial u_j}{\partial x_j}  + \rho_le_l\frac{D\phi_l}{Dt}\right\} = \sum_{l=1}^2 \frac{\partial \rho_l h_l a_{lj}}{\partial x_j}.
    \label{eq:entropy_int1}
\end{equation}
Substituting for $D\rho_l/Dt$ and $D\phi_l/Dt$ from Eqs. \eqref{eq:volumef}-\eqref{eq:massf} in Eq. \eqref{eq:entropy_int1}, we obtain an equation that relates $Ds_1/Dt$ and $Ds_2/Dt$ as
\begin{equation}
    \rho_1\phi_1T_1 \frac{Ds_1}{Dt} + \rho_2\phi_2T_2 \frac{Ds_2}{Dt} = \sum_{l=1}^2 \left\{\frac{\partial \rho_l h_l a_{lj}}{\partial x_j} - h_l\frac{\partial \rho_l a_{lj}}{\partial x_j} \right\}.
    \label{eq:entropy_rel1}
\end{equation}
However, this provides one relation for two unknowns, $Ds_1/Dt$ and $Ds_2/Dt$. Another relation can be obtained starting from the isobaric law and by taking the material derivative as
\begin{equation*}
    p_1(\rho_1, s_1) = p_2(\rho_2, s_2),
    \label{eq:entropy0}
\end{equation*}
\begin{equation*}
    \left.\frac{\partial p_1}{\partial \rho_1}\right\rvert_{s_1} \frac{D \rho_1}{D t} + \left.\frac{\partial p_1}{\partial s_1}\right\rvert_{\rho_1} \frac{D s_1}{D t} = \left.\frac{\partial p_2}{\partial \rho_2}\right\rvert_{s_2} \frac{D \rho_2}{D t} + \left.\frac{\partial p_2}{\partial s_2}\right\rvert_{\rho_2} \frac{D s_2}{D t}.
    \label{eq:entropy1}
\end{equation*}
Now, substituting for $D \rho_l/D t$, using the relations $\partial p_l/\partial \rho_l \rvert_{s_l}=c_l$ and $\left.\partial p_l/ \partial s_l \right\rvert_{\rho_l} = \rho_l \Gamma_l T_l$\textemdash
where $\Gamma_l$ is the Gruneisen coefficient for phase $l$ defined as $\Gamma_l = \left. (1/\rho_l) \partial p_l/\partial e_l\right\rvert_{\rho_l}$ which is equal to $\alpha_l$ for the EOS considered in this work\textemdash we arrive at the equation
\begin{equation}
\begin{aligned}
    \rho_1 \Gamma_1 T_1 \frac{D s_1}{D t} - \rho_2 \Gamma_2 T_2 \frac{D s_2}{D t} = \left(\frac{\partial \phi_1}{\partial t} + u_j\frac{\partial \phi_1}{\partial x_j} - \frac{\partial a_{1j}}{\partial x_j} \right) \left(\frac{\rho_1 c_1^2}{\phi_1} + \frac{\rho_2 c_2^2}{\phi_2} \right)\\
    + \frac{\partial u_j}{\partial x_j} \left(\rho_1 c_1^2 - \rho_2 c_2^2 \right) + a_{2j}\frac{\partial \rho_2}{\partial x_j} \left(\frac{c_2^2}{\phi_2} \right) - a_{1j}\frac{\partial \rho_1}{\partial x_j} \left(\frac{c_1^2}{\phi_1} \right).
\end{aligned}
\label{eq:entropy_int2}
\end{equation}
Now, subtracting the volume-fraction equation in Eq. \eqref{eq:volumef} from this equation, we arrive at the second relation for the two unknowns, $Ds_1/Dt$ and $Ds_2/Dt$, as
\begin{equation}
    \rho_1 \Gamma_1 T_1 \frac{D s_1}{D t} - \rho_2 \Gamma_2 T_2 \frac{D s_2}{D t} = a_{2j}\frac{\partial \rho_2}{\partial x_j} \left(\frac{c_2^2}{\phi_2} \right) - a_{1j}\frac{\partial \rho_1}{\partial x_j} \left(\frac{c_1^2}{\phi_1} \right).
    \label{eq:entropy_rel2}
\end{equation}
Now, solving for $Ds_1/Dt$ and $Ds_2/Dt$ using Eqs. \eqref{eq:entropy_rel1} and \eqref{eq:entropy_rel2}, we arrive at the transport equation for the individual phase entropy as
    \begin{equation}
    \begin{aligned}
    \rho_l\phi_l \frac{Ds_l}{Dt} = \frac{1}{T_l\left(\sum_{k=1}^2 \frac{\Gamma_k}{\phi_k} \right)}\left[
    a_{l'j}\frac{\partial \rho_{l'}}{\partial x_j} \left(\frac{c_{l'}^2}{\phi_{l'}} \right) 
    - a_{lj}\frac{\partial \rho_l}{\partial x_j} \left(\frac{c_l^2}{\phi_l} \right)\right.\\
    + \left.\frac{\Gamma_{l'}}{\phi_{l'}}\sum_{k=1}^2 \left\{\frac{\partial \rho_k h_k a_{kj}}{\partial x_j} - h_k\frac{\partial \rho_k a_{kj}}{\partial x_j} \right\}\right], 
    \end{aligned}
    \label{eq:entropy_rel}
\end{equation}
where the subscript $l'=1-l$ denotes another phase, different from phase $l$.
This relation in Eq. \eqref{eq:entropy_rel} shows that the material derivative of $s_l$ is not zero. The non-zero terms on the right-hand side is a function of interface-regularization flux $a_{li}$, which would go to zero if the internal structure of the interface is in a perfect equilibrium state\textemdash a hyperbolic-tangent function. Therefore, $Ds_l/Dt=0$ would be satisfied when the interface is in an equilibrium state. 
Now using the mass balance equation in Eq. \eqref{eq:massf} and rewriting Eq. \eqref{eq:entropy_rel} in a conservative form, a transport equation for the evolution of total entropy of phase $l$,  $\rho_l \phi_l s_l$, can be written as 
\begin{equation}
    \begin{aligned}
    \frac{\partial \rho_l \phi_l s_l}{\partial t} + \frac{\partial \rho_l u_j \phi_l s_l}{\partial x_j} = s_l \frac{\partial \rho_l a_{lj}}{\partial x_j}\\
    + \frac{1}{T_l\left(\sum_{k=1}^2 \frac{\Gamma_k}{\phi_k} \right)}\left[
    a_{l'j}\frac{\partial \rho_{l'}}{\partial x_j} \left(\frac{c_{l'}^2}{\phi_{l'}} \right) 
    - a_{lj}\frac{\partial \rho_l}{\partial x_j} \left(\frac{c_l^2}{\phi_l} \right)\right.\\
    + \left.\frac{\Gamma_{l'}}{\phi_{l'}}\sum_{k=1}^2 \left\{\frac{\partial \rho_k h_k a_{kj}}{\partial x_j} - h_k\frac{\partial \rho_k a_{kj}}{\partial x_j} \right\}\right].
    \end{aligned}
    \label{eq:phase_entropy}
\end{equation}
Here again, the right-hand side of Eq. \eqref{eq:phase_entropy} is a function of interface-regularization flux $a_{li}$, and hence, goes to zero when the interface is in an equilibrium state. Therefore, the entropy of each phase is not exactly conserved in a diffuse-interface method\textemdash even in the inviscid limit\textemdash due to the interface-regularization process which is an irreversible process (see Appendix A for the numerical quantification of entropy change associated with the interface-regularization process). 

However, an approximate conservation of entropy can be sought in the limit of equilibrium interface state ($a_{li}\rightarrow 0$). In this limit, we recover
\begin{equation*}
    \frac{\partial \rho_l \phi_l s_l}{\partial t} + \frac{\partial \rho_l u_j \phi_l s_l}{\partial x_j} = 0, \hspace{0.5cm} l=1,2,
\end{equation*}
and
\begin{equation}
    \frac{\partial \rho s}{\partial t} + \frac{\partial \rho u_j s}{\partial x_j} = 0,
    \label{eq:entropyf}
\end{equation}
where $\rho s = \sum_{l=1}^2 \rho_l \phi_l s_l$ is the total mixture entropy. With this, the total integrated entropy of each of the phases, $S_l=\int_{\Omega} \left(\rho_l \phi_l s_l\right) dV$, and the mixture entropy, $S=\int_{\Omega} \left(\rho s\right) dV$, are conserved, where $\Omega$ is the computational domain. 

\subsection{Reduced system for single-phase flows\label{sec:reduced}}
If the material properties such as density ($\rho_1=\rho_2=\rho$), viscosity ($\mu_1 = \mu_2 = \mu$), and parameters in the EOS ($\alpha_1=\alpha_2=\alpha; \beta_1=\beta_2=\beta$) are identical for the two phases, the two-phase system in Eqs. \eqref{eq:volumef}-\eqref{eq:energyf} will consistently reduce to the single-phase Navier-Stokes system, in the absence of surface tension. This can be written as
\begin{equation}
\frac{\partial \rho}{\partial t} + \frac{\partial \rho u_j}{\partial x_j} = 0,
\label{eq:masss}
\end{equation}
\begin{equation}
\frac{\partial \rho u_i}{\partial t} + \frac{\partial \rho u_i u_j}{\partial x_j} + \frac{\partial p}{\partial x_i} = 0,
\label{eq:moms}
\end{equation}
\begin{equation}
\frac{\partial E}{\partial t} + \frac{\partial \left(E + p \right) u_j}{\partial x_j} = 0,
\label{eq:energys}
\end{equation}
\begin{equation}
p = \alpha \rho e + \beta,
\label{eq:pressures}
\end{equation}
in the absence of viscous and gravity effects. Now, following the steps taken in Sections \ref{sec:ke_analytical}-\ref{sec:entropy_analytical}, the kinetic energy, internal energy, and entropy transport equations for the single-phase Navier-Stokes system can be written as

\begin{equation}
\frac{\partial }{\partial t}\left(\rho \frac{u_iu_i}{2}\right) + 
\frac{\partial}{\partial x_j} \left(\rho u_j \frac{u_i u_i}{2} \right)
+ \frac{\partial p u_j}{\partial x_j} - p \frac{\partial u_j}{\partial x_j}
= 0,
\label{eq:kinetics}
\end{equation}
\begin{equation}
    \frac{\partial \rho e}{\partial t} + \frac{\partial \rho e u_j}{\partial x_j} + 
    \frac{\partial p u_j}{\partial x_j} - u_j \frac{\partial p}{\partial x_j} = 0,
    \label{eq:ies}
\end{equation}
and
\begin{equation}
    \frac{\partial \rho s}{\partial t} + \frac{\partial \rho u_j s}{\partial x_j} = 0,
    \label{eq:entropys}
\end{equation}
respectively, in the absence of viscous and gravity effects.

\section{Numerical flux forms \label{sec:keep_scheme}}

In this section, we first derive the consistency conditions between the numerical fluxes that needs to be satisfied to achieve conservation of kinetic energy and entropy in Sections \ref{sec:ke_consistency}-\ref{sec:ie_consistency}. The consistency conditions are then summarized in Section \ref{sec:consistency_sum}. Finally, the proposed KEEP scheme, which consists of a set of numerical fluxes that satisfy the proposed consistency conditions, for single-phase and two-phase systems are presented in Section \ref{sec:keep}. 

Starting from the mass, momentum, and energy equations in Eqs. \eqref{eq:massf}-\eqref{eq:energyf}, 
and neglecting the viscous, surface tension, gravity, and thermal conduction terms; rewriting the total energy as $E=\rho(e+k)$; and utilizing the mixture mass equation in Eq. \eqref{eq:mod_continuity} instead of the individual mass equations in Eq. \eqref{eq:massf}, we can write the conservative system of equations in a general form as
\begin{equation}
\frac{\partial U_i}{\partial t} + \frac{\partial F_{ij}}{\partial x_j} + \frac{\partial P_{ij}}{\partial x_j} = \frac{\partial R_{ij}}{\partial x_j},
\label{eq:reduced_system}
\end{equation}
where $U_i$ is the state vector for the conserved variables and $F_{ij}$, $P_{ij}$, and $R_{ij}$ are the flux vectors for the convective, pressure, and interface-regularization terms, respectively, and are given by
\begin{equation}
\begin{gathered}
U_i=
\begin{bmatrix}
\rho\\
\rho u_i\\
\rho e + \rho k
\end{bmatrix},
F_{ij}=
\begin{bmatrix}
\rho u_j\\
\rho u_i u_j\\
\rho e u_j + \rho k u_j
\end{bmatrix},\\
P_{ij}=
\begin{bmatrix}
0\\
p \delta_{ij}\\
p u_j
\end{bmatrix},
\mathrm{and}\ 
R_{ij}=
\begin{bmatrix}
\sum_{l=1}^2 \rho_l a_{lj} = f_j\\
u_i f_j\\
k f_j + \sum_{l=1}^2 \rho_l h_l a_{lj}
\end{bmatrix}.
\end{gathered}
\label{eq:flux_vec}
\end{equation}

Adopting the finite-volume approach, the system of equations can be discretized in space using the conservative numerical fluxes of the form
\begin{equation}
    \left.\frac{\partial q_j}{\partial x_j}\right\vert_m \approx \frac{\hat{q}_j\rvert_{(m+\frac{1}{2})} - \hat{q}_j\rvert_{(m-\frac{1}{2})}}{\Delta x_j},
    \label{eq:gen_disc}
\end{equation}
where $\hat{q}_j\rvert_{(m\pm 1/2)}$ represents the numerical flux for a generic flux function $q_j$ evaluated on the cell face, subscript $m$ denotes the grid-cell index, and $\Delta x_j$ represents the grid-cell size. Now, the semi-discrete representation for the system of equations can be written as
\begin{equation}
\left.\frac{\partial U_i}{\partial t}\right\vert_m + \frac{ \hat{F}_{ij}\rvert_{(m+\frac{1}{2})} - \hat{F}_{ij}\rvert_{(m-\frac{1}{2})}}{\Delta x_j} + \frac{ \hat{P}_{ij}\rvert_{(m+\frac{1}{2})} - \hat{P}_{ij}\rvert_{(m-\frac{1}{2})}}{\Delta x_j} = \frac{ \hat{R}_{ij}\rvert_{(m+\frac{1}{2})} - \hat{R}_{ij}\rvert_{(m-\frac{1}{2})}}{\Delta x_j},
\label{eq:semi_discrete_system}
\end{equation}
where $\hat{F}_{ij}$, $\hat{P}_{ij}$, and $\hat{R}_{ij}$ are the numerical flux vectors for the convective, pressure, and interface-regularization terms, respectively. The components of the numerical flux vectors can be written as
\begin{equation}
\begin{gathered}
\hat{F}_{ij}\rvert_{(m\pm \frac{1}{2})}=
\begin{bmatrix}
\reallywidehat{\rho u_j}\rvert_{(m\pm \frac{1}{2})} = \hat{C}_j\rvert_{(m\pm \frac{1}{2})}\\
\reallywidehat{\rho u_i u_j}\rvert_{(m\pm \frac{1}{2})} = \hat{M}_{ij}\rvert_{(m\pm \frac{1}{2})}\\
\reallywidehat{\rho e u_j}\rvert_{(m\pm \frac{1}{2})} + \reallywidehat{\rho k u_j}\rvert_{(m\pm \frac{1}{2})} = \hat{I}_j\rvert_{(m\pm \frac{1}{2})} + \hat{K}_j\rvert_{(m\pm \frac{1}{2})}
\end{bmatrix},\\
\hat{P}_{ij}\rvert_{(m\pm \frac{1}{2})}=
\begin{bmatrix}
0\\
\reallywidehat{p \delta_{ij}}\rvert_{(m\pm \frac{1}{2})} = \hat{P}_{ij}\rvert_{(m\pm \frac{1}{2})}\\
\reallywidehat{p u_j}\rvert_{(m\pm \frac{1}{2})} = \hat{W}_j\rvert_{(m\pm \frac{1}{2})}
\end{bmatrix},\\
\mathrm{and}\ 
\hat{R}_{ij}\rvert_{(m\pm \frac{1}{2})}=
\begin{bmatrix}
\reallywidehat{f_j}\rvert_{(m\pm \frac{1}{2})} = \hat{F}_j\rvert_{(m\pm \frac{1}{2})}\\
\reallywidehat{u_i f_j}\rvert_{(m\pm \frac{1}{2})} = \hat{R}_{ij}\rvert_{(m\pm \frac{1}{2})}\\
\reallywidehat{k f_j}\rvert_{(m\pm \frac{1}{2})} + \reallywidehat{\sum_{l=1}^2 \rho_l h_l a_{lj}}\rvert_{(m\pm \frac{1}{2})} = \hat{T}_j\rvert_{(m\pm \frac{1}{2})} + \hat{H}_j\rvert_{(m\pm \frac{1}{2})}
\end{bmatrix}.
\end{gathered}
\label{eq:num_flux_vec}
\end{equation}
Using any form of numerical flux for the components in Eq. \eqref{eq:num_flux_vec} leads to an implementation that discretely conserves the variables in the state vector $U_i$ in Eq. \eqref{eq:semi_discrete_system} due to the telescoping property. However, to achieve secondary conservation of the kinetic energy and entropy, the numerical fluxes need to be approximated in a consistent manner, such that the conservation can be achieved implicitly.

In addition to these fluxes, the convective numerical fluxes in the volume fraction equation [Eq. \eqref{eq:volumef}] and individual phase mass equation [Eq. \eqref{eq:massf}] can be written as
\begin{equation}
    \reallywidehat{\phi u_j}\rvert_{(m\pm \frac{1}{2})} = \hat{\Phi}_j\rvert_{(m\pm \frac{1}{2})},
    \label{eq:vol_conv}
\end{equation}
and 
\begin{equation}
    \reallywidehat{m_l u_j}\rvert_{(m\pm \frac{1}{2})} = \hat{Q}_{lj}\rvert_{(m\pm \frac{1}{2})},
\end{equation}
respectively; and the interface-regularization fluxes in the volume fraction equation and individual phase mass equation can be written as
\begin{equation}
    \reallywidehat{a_{lj}}\rvert_{(m\pm \frac{1}{2})} = \hat{A}_{lj}\rvert_{(m\pm \frac{1}{2})},
\end{equation}
and 
\begin{equation}
    \reallywidehat{\rho_l a_{lj}}\rvert_{(m\pm \frac{1}{2})} = \hat{B}_{lj}\rvert_{(m\pm \frac{1}{2})},
    \label{eq:mass_reg}
\end{equation}
A consistent numerical approximation for these fluxes in Eqs. \eqref{eq:vol_conv}-\eqref{eq:mass_reg} is also crucial in achieving secondary conservation of the kinetic energy and entropy.

\subsection{Skew-symmetric-like splitting}

The derivatives of the non-linear products that typically arise in the convective terms can be discretized in different ways. If the generic non-linear flux in Eq. \eqref{eq:gen_disc} can be expressed as $q_1 = a b c$, where $a$, $b$, and $c$ are functions of $x$, then according to Eq. \eqref{eq:gen_disc}, the three possible numerical flux forms for $\hat{q}_1\rvert_{(m\pm 1/2)}$ are
\begin{equation}
    \begin{gathered}
        \hat{q}_1\rvert_{(m\pm\frac{1}{2})} = \xbar{abc}^{(m\pm\frac{1}{2})},\\
        \hat{q}_1\rvert_{(m\pm\frac{1}{2})} = \xbar{ab}^{(m\pm\frac{1}{2})} \xbar{c}^{(m\pm\frac{1}{2})},\\
        \hat{q}_1\rvert_{(m\pm\frac{1}{2})} = \xbar{a}^{(m\pm\frac{1}{2})} \xbar{b}^{(m\pm\frac{1}{2})} \xbar{c}^{(m\pm\frac{1}{2})},\\
    \end{gathered}
    \label{eq:split_form}
\end{equation}
which are called the divergence form, quadratic-split form, and cubic-split form, respectively, and the overbar denotes an arithmetic average of quantity evaluated at $m$ and $m\pm1$, e.g., $\xbar{ab}^{(m\pm\frac{1}{2})}$ is
\begin{equation}
    \xbar{ab}^{(m\pm\frac{1}{2})} = \frac{ab\rvert_{m}+ab\rvert_{(m\pm1)}}{2}.
\end{equation}
However, the interface-regularization flux terms contain non-linear product of four or five quantities, and therefore, new numerical split forms need to be defined. If $q_1 = a b c d$, then a quartic-split form of the numerical flux can be defined as
\begin{equation}
     \hat{q}_1\rvert_{(m\pm\frac{1}{2})} = \xbar{a}^{(m\pm\frac{1}{2})} \xbar{b}^{(m\pm\frac{1}{2})} \xbar{c}^{(m\pm\frac{1}{2})} \xbar{d}^{(m\pm\frac{1}{2})},
     \label{eq:quartic_split_form}
\end{equation}
Similarly, if $q_1 = a b c d e$, then a quintic-split form of the numerical flux can be defined as
\begin{equation}
     \hat{q}_1\rvert_{(m\pm\frac{1}{2})} = \xbar{a}^{(m\pm\frac{1}{2})} \xbar{b}^{(m\pm\frac{1}{2})} \xbar{c}^{(m\pm\frac{1}{2})} \xbar{d}^{(m\pm\frac{1}{2})} \xbar{e}^{(m\pm\frac{1}{2})},
     \label{eq:quintic_split_form}
\end{equation}
These skew-symmetric-like split schemes are used here, because these schemes result in reduced aliasing errors \citep{blaisdell1996effect,chow2003further,kennedy2008reduced} and are known to improve the conservation properties of the quadratic quantities \citep{kravchenko1997effect}.

The numerical fluxes, introduced here, can also be equivalently rewritten in a pointwise manner, e.g., the divergence form, quadratic-split form, and cubic-split forms of the numerical flux can be written as
\begin{equation}
    \begin{gathered}
        \frac{\partial q_1}{\partial x_1} \approx \left.\frac{\delta a b c}{\delta x}\right|_{m},\\
        \frac{\partial q_1}{\partial x_1} \approx \frac{1}{2} \left( \left.\frac{\delta a b c}{\delta x}\right|_{m} + \left.c\right|_{m} \left.\frac{\delta a b }{\delta x}\right|_{m} + \left.a b \right|_{m} \left.\frac{\delta c}{\delta x}\right|_{m} \right),\\
        \frac{\partial q_1}{\partial x_1} \approx \frac{1}{4} \left( \left.\frac{\delta a b c}{\delta x}\right|_{m} + \left.c\right|_{m} \left.\frac{\delta a b }{\delta x}\right|_{m} + \left.a\right|_{m} \left.\frac{\delta b c }{\delta x}\right|_{m} +\left. b\right|_{m} \left.\frac{\delta a c }{\delta x}\right|_{m} \right.\\ 
        \left. + \left.a b \right|_{m} \left.\frac{\delta c}{\delta x}\right|_{m}
        + \left.b c \right|_{m} \left.\frac{\delta a}{\delta x}\right|_{m}
        + \left.a c \right|_{m}\left.\frac{\delta b}{\delta x}\right|_{m}
        \right),
    \end{gathered}
\end{equation}
respectively, where the $\delta/ \delta x$ operator is the standard finite-difference derivative.

\subsection{Consistency conditions for volume fraction, mass, momentum, and kinetic-energy fluxes \label{sec:ke_consistency}}

The kinetic energy equations in Eqs. \eqref{eq:kineticf}, \eqref{eq:kinetics} are not solved directly, and hence, we seek an implicit conservation of kinetic energy here. One approach to achieve this is by following the same steps taken in Section \ref{sec:ke_analytical} discretely, and by seeking those forms of numerical fluxes in the mass, momentum, and energy equations that result in the implied discrete kinetic energy equation that has discrete divergence operators in all the terms. In this case, the contribution of the terms in the discrete kinetic energy equation will be only through the boundary terms or zero, in a periodic domain, due to the telescoping property of the divergence operator. Note that we are seeking kinetic energy conservation in the limit of zero viscosity, surface tension, gravity force, and reversible pressure work. 

Theorem \ref{theorem:global_consistency} proposes consistency conditions between
$\hat{C}_j$ and $\hat{M}_{ij}$, and
$\hat{F}_j$ and $\hat{R}_{ij}$, such that the discrete convective and interface-regularization terms do not contribute spuriously to the total kinetic energy. 

\newtheorem{theorem}{Theorem}[section]

\begin{theorem}
The total kinetic energy of the system can be discretely conserved in the absence of reversible pressure work, viscosity, surface tension, gravity, and time-differencing errors if the following consistency conditions are satisfied: those between the convective fluxes of mass $\hat{C}$ and momentum $\hat{M}$, 
\begin{equation}
    \hat{M}_{ij}\rvert_{(m\pm\frac{1}{2})} = \hat{C}_j\rvert_{(m\pm\frac{1}{2})}
    \xbar{u_i}^{(m\pm\frac{1}{2})}
    \label{eq:mom_cons}
\end{equation}
and between the interface-regularization fluxes of mass $\hat{F}$ and momentum $\hat{R}$,  
\begin{equation}
    \hat{R}_{ij}\rvert_{(m\pm\frac{1}{2})} = \hat{F}_j\rvert_{(m\pm\frac{1}{2})}
    \xbar{u_i}^{(m\pm\frac{1}{2})}
    \label{eq:int_mom_cons}
\end{equation}
where $u$ is the velocity, $m$ is the grid index, and $i$ and $j$ are the indices in the Einstein notation.

\label{theorem:global_consistency}
\end{theorem}


\begin{proof}


Ignoring the time-stepping errors, we can obtain the semi-discrete kinetic energy equation by multiplying the semi-discrete momentum equation by $u_i$
\begin{equation}
u_i\frac{\partial \rho u_i}{\partial t} + u_i\rvert_m\left(\frac{\hat{M}_{ij}\rvert_{(m+\frac{1}{2})} - \hat{M}_{ij}\rvert_{(m-\frac{1}{2})}}{\Delta x_j}\right) -
u_i\rvert_m\left(\frac{\hat{R}_{ij}\rvert_{(m+\frac{1}{2})} - \hat{R}_{ij}\rvert_{(m-\frac{1}{2})}}{\Delta x_j}\right) +
u_i (\mathrm{O.T.}) = 0,
\label{eq:discrete_u_dot_mom}
\end{equation}
where $\mathrm{O.T.}$ represents other terms such as pressure, viscosity, surface tension, and gravity force. Here, Eq. \eqref{eq:discrete_u_dot_mom} is a discrete analogue of Eq. \eqref{eq:u_dot_mom}.
Now, suppose steps taken in going from Eq. \eqref{eq:u_dot_mom} to Eq. \eqref{eq:intermediate_kinetic} are valid discretely, then, discretely following the steps in going from Eq. \eqref{eq:intermediate_kinetic} to Eq. \eqref{eq:kineticf}, i.e., subtracting out the continuity equation contribution multiplied by $u_iu_i/2$ term from this equation, we will arrive at the discrete analogue of the kinetic energy equation in Eq. \ref{eq:kineticf} with the required form of the discrete divergence operators for the convective and interface-regularization terms. This is sufficient to show
\begin{equation*}
    \sum_m \frac{\partial \rho k}{\partial t}=0.
\end{equation*}

Therefore, the required condition to achieve discrete kinetic energy conservation is that the steps taken in going from Eq. \eqref{eq:u_dot_mom} to Eq. \eqref{eq:intermediate_kinetic} be valid discretely.
The discrete analogue of Eq. \eqref{eq:intermediate_kinetic} can be written as 
\begin{equation}
\begin{aligned}
\frac{\partial \rho k}{\partial t} +
\frac{1}{2}u_iu_i\rvert_m\frac{\partial \rho}{\partial t} +
\frac{1}{2}u_iu_i\rvert_m\left(\frac{\hat{C}_{j}\rvert_{(m+\frac{1}{2})} - \hat{C}_{j}\rvert_{(m-\frac{1}{2})}}{\Delta x_j}\right) -
\frac{1}{2}u_iu_i\rvert_m\left(\frac{\hat{F}_{j}\rvert_{(m+\frac{1}{2})} - \hat{F}_{j}\rvert_{(m-\frac{1}{2})}}{\Delta x_j}\right)\\
+\left(\frac{\hat{K}_{j}\rvert_{(m+\frac{1}{2})} - \hat{K}_{j}\rvert_{(m-\frac{1}{2})}}{\Delta x_j}\right) -
\left(\frac{\hat{T}_{j}\rvert_{(m+\frac{1}{2})} - \hat{T}_{j}\rvert_{(m-\frac{1}{2})}}{\Delta x_j}\right) +
u_i(\mathrm{O.T.}) = 0.
\end{aligned}
\label{eq:discrete_intermediate_kinetic}
\end{equation}
Therefore, Eq. \eqref{eq:discrete_intermediate_kinetic} should be equivalent to Eq. \eqref{eq:discrete_u_dot_mom} to achieve discrete kinetic energy conservation.
Seeking global conservation of kinetic energy, we equate Eq. \eqref{eq:discrete_u_dot_mom} and Eq. \eqref{eq:discrete_intermediate_kinetic} under the summation over all grid cells as
\begin{equation}
\begin{aligned}
    \sum_m \left\{u_i\rvert_m\left(\frac{\hat{M}_{ij}\rvert_{(m+\frac{1}{2})} - \hat{M}_{ij}\rvert_{(m-\frac{1}{2})}}{\Delta x_j}\right) -
u_i\rvert_m\left(\frac{\hat{R}_{ij}\rvert_{(m+\frac{1}{2})} - \hat{R}_{ij}\rvert_{(m-\frac{1}{2})}}{\Delta x_j}\right)\right\}\\
= \sum_m \Bigg\{\frac{1}{2}u_iu_i\rvert_m\left(\frac{\hat{C}_{j}\rvert_{(m+\frac{1}{2})} - \hat{C}_{j}\rvert_{(m-\frac{1}{2})}}{\Delta x_j}\right) -
\frac{1}{2}u_iu_i\rvert_m\left(\frac{\hat{F}_{j}\rvert_{(m+\frac{1}{2})} - \hat{F}_{j}\rvert_{(m-\frac{1}{2})}}{\Delta x_j}\right)\\
+\left(\frac{\hat{K}_{j}\rvert_{(m+\frac{1}{2})} - \hat{K}_{j}\rvert_{(m-\frac{1}{2})}}{\Delta x_j}\right) -\left(\frac{\hat{T}_{j}\rvert_{(m+\frac{1}{2})} - \hat{T}_{j}\rvert_{(m-\frac{1}{2})}}{\Delta x_j}\right)\Bigg\}.
\end{aligned}
\end{equation}
$\hat{K}_j\rvert_{(m\pm\frac{1}{2})}$ and $\hat{T}_j\rvert_{(m\pm\frac{1}{2})}$ terms form telescoping series and therefore their contribution cancels out. Expanding the summation and separating the contributions from the convective and interface-regularization terms, we arrive at the conditions
\begin{equation}
\hat{M}_j\rvert_{(m\pm\frac{1}{2})} = \hat{C}_j\rvert_{(m\pm\frac{1}{2})}
\xbar{u_i}^{(m\pm\frac{1}{2})},
\label{eq:mom_cons_proof}
\end{equation}
and 
\begin{equation}
    \hat{R}_j\rvert_{(m\pm\frac{1}{2})} = \hat{F}_j\rvert_{(m\pm\frac{1}{2})}
    \xbar{u_i}^{(m\pm\frac{1}{2})},
    \label{eq:int_mom_cons_proof}
\end{equation}
which concludes the proof.
\end{proof}

The consistency condition in Eq. \eqref{eq:mom_cons} is similar to the kinetic energy preserving (KEP) flux by \citet{jameson2008formulation}, which is a sufficient condition for numerical stability of incompressible flows because of the absence of exchange of energy between kinetic energy and internal energy, unlike for compressible flows. For compressible flows, an additional constraint on the entropy of the system is needed to achieve stable numerical simulations.

There are many approaches to achieve (exact or approximate) discrete conservation of entropy. One option would be to directly solve a transport equation for entropy instead of an energy equation, which was explored in \citet{honein2005}. They also proposed a modification to the total energy equation, in \citet{honein2004higher}, which would result in higher conservation of entropy ($\rho s$ and $\rho s^2$) implicitly. Another option would be to solve an internal energy equation or a total energy equation with consistent numerical discretization that results in conservation of entropy implicitly. This approach has been previously used by \citet{chandrashekar2013kinetic}, \citet{kuya2018kinetic}, and \citet{coppola2019numerically}. This approach that involves solving an energy equation with implicit conservation of entropy is a more preferred option since it also results in the primary conservation of energy, which is crucial for compressible flows. Therefore, we choose to adopt the second option in this work.

Although, the consistency conditions in Theorem \ref{theorem:global_consistency} result in discretely conserving the total kinetic energy, in the absence of reversible pressure work, this is not a sufficient condition for numerical stability of compressible flows. 
This is because an incorrect local exchange of energy between kinetic energy and internal energy could spuriously contribute to entropy of the system.
Hence, the first step toward achieving conservation of entropy implicitly is to propose additional constraints on the transport of kinetic energy, such that there is no spurious local exchange of energy between kinetic energy and internal energy (a ``local" kinetic energy preservation).  
To this end, Theorem \ref{theorem:local_consistency} proposes consistency conditions between
$\hat{C}_j$ and $\hat{K}_{j}$, and 
$\hat{F}_j$ and $\hat{T}_{j}$, such that the discrete convective and interface-regularization terms do not contribute spuriously to the local transport of kinetic energy.

\begin{theorem}
The discrete kinetic energy of the system can be locally conserved in the absence of reversible pressure work, viscosity, surface tension, gravity, and time-differencing errors if the following consistency conditions are satisfied: 
those between the convective fluxes of mass $\hat{C}$ and kinetic energy $\hat{K}$, 
\begin{equation}
    \hat{K}_j\rvert_{(m\pm\frac{1}{2})} = \hat{C}_j\rvert_{(m\pm\frac{1}{2})}
    \frac{u_i\rvert_{(m\pm1)} u_i\rvert_{m}}{2},
    \label{eq:ke_cons}
\end{equation}
and between the interface-regularization fluxes of mass $\hat{F}$ and kinetic energy $\hat{T}$,  
\begin{equation}
    \hat{T}_j\rvert_{(m\pm\frac{1}{2})} = \hat{F}_j\rvert_{(m\pm\frac{1}{2})} \frac{u_i\rvert_{(m\pm1)} u_i\rvert_{m}}{2},
    \label{eq:int_ke_cons}
\end{equation}
where $u$ is the velocity, $m$ is the grid index, and $i$ and $j$ are the indices in the Einstein notation.

\label{theorem:local_consistency}
\end{theorem}

\begin{proof}
Following the proof of Theorem \ref{theorem:global_consistency}, and now instead seeking local conservation of kinetic energy, we equate Eq. \eqref{eq:discrete_u_dot_mom} and Eq. \eqref{eq:discrete_intermediate_kinetic} in a pointwise manner as
\begin{equation}
\begin{aligned}
    u_i\rvert_m\left(\frac{\hat{M}_{ij}\rvert_{(m+\frac{1}{2})} - \hat{M}_{ij}\rvert_{(m-\frac{1}{2})}}{\Delta x_j}\right) -
u_i\rvert_m\left(\frac{\hat{R}_{ij}\rvert_{(m+\frac{1}{2})} - \hat{R}_{ij}\rvert_{(m-\frac{1}{2})}}{\Delta x_j}\right)\\
= \frac{1}{2}u_iu_i\rvert_m\left(\frac{\hat{C}_{j}\rvert_{(m+\frac{1}{2})} - \hat{C}_{j}\rvert_{(m-\frac{1}{2})}}{\Delta x_j}\right) -
\frac{1}{2}u_iu_i\rvert_m\left(\frac{\hat{F}_{j}\rvert_{(m+\frac{1}{2})} - \hat{F}_{j}\rvert_{(m-\frac{1}{2})}}{\Delta x_j}\right)\\
+\left(\frac{\hat{K}_{j}\rvert_{(m+\frac{1}{2})} - \hat{K}_{j}\rvert_{(m-\frac{1}{2})}}{\Delta x_j}\right) -\left(\frac{\hat{T}_{j}\rvert_{(m+\frac{1}{2})} - \hat{T}_{j}\rvert_{(m-\frac{1}{2})}}{\Delta x_j}\right).
\end{aligned}
\end{equation}
Making use of the conditions in Eqs. \eqref{eq:mom_cons_proof} and \eqref{eq:int_mom_cons_proof}, and separating the contributions from the convective and interface-regularization terms, we arrive at the conditions
\begin{equation}
    \hat{K}_j\rvert_{(m\pm\frac{1}{2})} = \hat{C}_j\rvert_{(m\pm\frac{1}{2})} \frac{u_i\rvert_{(m\pm1)} u_i\rvert_{m}}{2},
    \label{eq:ke_cons_proof}
\end{equation}
and
\begin{equation}
    \hat{T}_j\rvert_{(m\pm\frac{1}{2})} = \hat{F}_j\rvert_{(m\pm\frac{1}{2})} \frac{u_i\rvert_{(m\pm1)} u_i\rvert_{m}}{2},
    \label{eq:int_ke_cons_proof}
\end{equation}
which concludes the proof.
\end{proof}

When individual mass balance equations in Eq. \eqref{eq:massf} are solved for each of the phases, they are supposed to be solved consistently with other equations. Hence, additional consistency conditions can be derived for these. 
Note that the mass balance equations for phases $1$ and $2$ were summed to obtain the total mixture mass balance equation in Eq. \eqref{eq:mod_continuity}, and therefore, the same criterion should hold for numerical fluxes, e.g., as 
\begin{equation}
    \sum_{l=1}^2 \left.\frac{\delta u_j m_l}{\delta x_j}\right\vert_m = \left.\frac{\delta u_j \rho}{\delta x_j}\right\vert_m.
\end{equation}
This results in the consistency condition between $\hat{C}_j$ and $\hat{Q}_{lj}$ as
\begin{equation}
    \sum_{l=1}^2 \hat{Q}_{lj}\rvert_{(m\pm 1/2)} = \hat{C}_j\rvert_{(m\pm 1/2)}.
\end{equation}
Similarly, the consistency condition between $\hat{F}_j$ and $\hat{B}_{lj}$ can be written as
\begin{equation}
    \sum_{l=1}^2 \hat{B}_{lj}\rvert_{(m\pm 1/2)} = \hat{F}_j\rvert_{(m\pm 1/2)}.
\end{equation}
In the limit of incompressibility, mass balance equations in Eq. \eqref{eq:massf} reduce to the volume fraction equation in Eq. \eqref{eq:volumef} \citep{jain2020conservative}. This can be used to obtain a consistency condition between $\hat{C}_j$ and $\hat{\Phi}_j$. The consistency condition states that the split-numerical flux form used in $\hat{C}_j\rvert_{(m\pm\frac{1}{2})}$ should be the same as the one used in $\hat{\Phi}_j\rvert_{(m\pm\frac{1}{2})}$.  






\subsection{Consistency conditions for internal energy, pressure, and enthalpy fluxes \label{sec:ie_consistency}}

In addition to satisfying the local kinetic energy preservation in Theorem \ref{theorem:local_consistency}, in order to achieve implicit conservation of entropy without having to solve for Eqs. \eqref{eq:phase_entropy} or \eqref{eq:entropyf}, all the steps taken in Section \ref{sec:entropy_analytical} in arriving at Eq. \eqref{eq:entropyf} starting from the internal energy in Eq, \eqref{eq:ief} need to hold discretely. The mass balance equation and the volume-fraction advection equations used in Eqs. \eqref{eq:entropy_int1} and Eq. \eqref{eq:entropy_int2} are directly solved and hence are satisfied discretely. The only remaining constraint is that the method needs to satisfy the interface-equilibrium condition (IEC).  


According to \citet{abgrall1996prevent}, if $u^k\rvert_m=u_0$ and $p^k\rvert_m = p_0$, where $k$ is the time-step index, any model or a numerical scheme that satisfies $u\rvert_m^{k+1}=u_0$ and $p\rvert_m^{k+1}=p_0$, $\forall m$, is said to satisfy the IEC. It is known that the IEC needs to be satisfied to avoid spurious solutions and to achieve stable numerical solutions. \citet{jain2020conservative} derived the analytical formulation of the interface-regularization terms in the conservative diffuse-interface method such that the resulting internal energy equation in Eq. \eqref{eq:ief} satisfies the IEC, and showed that this results in approximate conservation of entropy (analytically). 
We here, seek fluxes that verify discretely the IEC which results in improved discrete entropy preservation. To do so, we formulate consistency conditions between $\hat{I}_j$ and $\hat{C}_j$, and $\hat{H}_j$ and $\hat{F}_j$ in Lemma \ref{lem:IEC_consistency} and the proof is presented in Appendix B. 

\newtheorem{lemma1}{Lemma}[section]

\begin{lemma1}
The consistency conditions between $\hat{I}_j$ and $\hat{C}_j$, and $\hat{H}_j$ and $\hat{F}_j$, required to satisfy IEC, can be stated as: the split-numerical flux form used in $\hat{I}_j\rvert_{(m\pm\frac{1}{2})}$ should be the same as the one in $\hat{C}_j\rvert_{(m\pm\frac{1}{2})}$, and the split-numerical flux form used in $\hat{H}_j\rvert_{(m\pm\frac{1}{2})}$ should be the same as the one in $\hat{F}_j\rvert_{(m\pm\frac{1}{2})}$.
\label{lem:IEC_consistency}
\end{lemma1}



%


Finally, for the pressure flux terms, a discrete consistency condition can be derived as follows. The sum of pressure work in the transport equation for kinetic energy in Eq. \eqref{eq:kineticf} and the pressure dilatation term in the transport equation for the internal energy in Eq. \eqref{eq:ief} gives the conservative pressure term in the total energy equation in Eq. \eqref{eq:energyf}, an idea that was also used by \citet{kok2009high,kuya2018kinetic}. This condition should be satisfied discretely, and can be written as
\begin{equation*}
        \left.u_j\frac{\delta p}{\delta x_j}\right\vert_m + \left.p \frac{\delta u_j}{\delta x_j}\right\vert_m = \left.\frac{\delta p u_j}{\delta x_j}\right\vert_m
\end{equation*}
Writing this in terms of fluxes, a consistency condition between $\hat{W}_j$ and $\hat{P}_{ij}$ is given by
\begin{equation*}
    \hat{P}_{ij}\rvert_{(m\pm\frac{1}{2})} \xbar{u_i}^{(m\pm\frac{1}{2})} + p\rvert_{m} \xbar{u_j}^{(m\pm\frac{1}{2})} = \hat{W}_j\rvert_{(m\pm\frac{1}{2})}.
\end{equation*}


\subsection{Summary of the consistency conditions\label{sec:consistency_sum}}

The consistency conditions used in this work can be summarized as follows:
\begin{enumerate}
    \item The kinetic energy in the system is discretely conserved, globally and locally, in the absence of reversible pressure work, dissipative mechanisms, and time-differencing errors, if the following relations are satisfied 
    \begin{equation*}
    \hat{M}_{ij}\rvert_{(m\pm\frac{1}{2})} = \hat{C}_j\rvert_{(m\pm\frac{1}{2})}
    \xbar{u_i}^{(m\pm\frac{1}{2})},
    \end{equation*}
    \begin{equation*}
    \hat{K}_j\rvert_{(m\pm\frac{1}{2})} = \hat{C}_j\rvert_{(m\pm\frac{1}{2})}
    \frac{u_i\rvert_{(m\pm1)} u_i\rvert_{m}}{2},
    \end{equation*}
    \begin{equation*}
    \hat{R}_{ij}\rvert_{(m\pm\frac{1}{2})} = \hat{F}_j\rvert_{(m\pm\frac{1}{2})}
    \xbar{u_i}^{(m\pm\frac{1}{2})},
    \end{equation*}
    \begin{equation*}
    \hat{T}_j\rvert_{(m\pm\frac{1}{2})} = \hat{F}_j\rvert_{(m\pm\frac{1}{2})} \frac{u_i\rvert_{(m\pm1)} u_i\rvert_{m}}{2}.
    \end{equation*}
    \item The mixture mass flux can be obtained by summing the mass fluxes of the individual phases, which results in the relations
    \begin{equation*}
    \sum_{l=1}^2 \hat{Q}_{lj}\rvert_{(m\pm 1/2)} = \hat{C}_j\rvert_{(m\pm 1/2)},
    \end{equation*}
    \begin{equation*}
    \sum_{l=1}^2 \hat{B}_{lj}\rvert_{(m\pm 1/2)} = \hat{F}_j\rvert_{(m\pm 1/2)}.
    \end{equation*}
    \item The split-numerical flux form used in $\hat{C}_j\rvert_{(m\pm\frac{1}{2})}$ should be the same as the ones used in $\hat{\Phi}_j\rvert_{(m\pm\frac{1}{2})}$ and $\hat{I}_j\rvert_{(m\pm\frac{1}{2})}$ to satisfy IEC.
     \item The split-numerical flux form used in $\hat{F}_j\rvert_{(m\pm\frac{1}{2})}$ should be the same as the one used in $\hat{H}_j\rvert_{(m\pm\frac{1}{2})}$ to satisfy IEC.
     \item The sum of pressure work in the kinetic energy equation and the pressure dilatation term in the internal energy equation should be equal to the conservative pressure-diffusion term in the total energy equation. This results in the relation
    \begin{equation*}
    \hat{P}_{ij}\rvert_{(m\pm\frac{1}{2})} \xbar{u_i}^{(m\pm\frac{1}{2})} + p\rvert_{m} \xbar{u_j}^{(m\pm\frac{1}{2})} = \hat{W}_j\rvert_{(m\pm\frac{1}{2})}.
    \end{equation*}
\end{enumerate}

The consistency conditions $1$\textendash $5$ completely define the numerical scheme, provided $\hat{C}_j$ and $\hat{F}_j$ are chosen. Note that for the single-phase Navier-Stokes system, the consistency conditions in $1$, $3$, and $5$ are sufficient to completely define the numerical scheme, provided $\hat{C}_j$ is chosen.

\subsection{The proposed KEEP scheme \label{sec:keep}}

The convective mass flux is a product of $\rho$ and $u_j$. Hence, using the quadratic-split form defined in Eq. \eqref{eq:split_form} for $\hat{C}_j$, we can derive the numerical fluxes that result in a consistent set of numerical discretization for the single-phase Navier-Stokes system, using the consistency conditions in Section \ref{sec:consistency_sum}, as 
\begin{equation}
\begin{aligned}
    &\hat{C}_j\rvert_{(m\pm\frac{1}{2})} = \xbar{\rho}^{(m\pm\frac{1}{2})} \xbar{u}_{j}^{(m\pm\frac{1}{2})},\\
    &\hat{M}_{ij}\rvert_{(m\pm\frac{1}{2})} = \xbar{\rho}^{(m\pm\frac{1}{2})} \xbar{u}_j^{(m\pm\frac{1}{2})} \xbar{u}_i^{(m\pm\frac{1}{2})},\\
    &\hat{K}_j\rvert_{(m\pm\frac{1}{2})} = \xbar{\rho}^{(m\pm\frac{1}{2})} 
    \xbar{u}_j^{(m\pm\frac{1}{2})} \frac{u_i\rvert_{(m\pm1)} u_i\rvert_{m}}{2},\\
    &\hat{I}_j\rvert_{(m\pm\frac{1}{2})} =  (\xbar{\rho e})^{(m\pm\frac{1}{2})} \xbar{u}_j^{(m\pm\frac{1}{2})},\\
    &\hat{P}_{ij}\rvert_{(m\pm \frac{1}{2})} = \xbar{p}^{(m\pm\frac{1}{2})} \delta_{ij},\\
    &\hat{W}_j\rvert_{(m\pm\frac{1}{2})} =  \frac{u_j\rvert_{(m\pm1)} p\rvert_{m} + u_j\rvert_{m} p\rvert_{(m\pm1)}}{2}.
\end{aligned}
\label{eq:single_phase_keep_scheme}
\end{equation}
The mass-regularization flux in Eq. \eqref{eq:mod_continuity} consists of a diffusive flux term and a non-linear sharpening flux term, that can be written as
\begin{equation*}
    \hat{F}_j\rvert_{(m\pm \frac{1}{2})} = \sum_{l=1}^2 \reallywidehat{\rho_l a_{lj}}\rvert_{(m\pm\frac{1}{2})} = \Gamma \sum_{l=1}^2 \left.\left\{\epsilon\underbrace{\left(\reallywidehat{\rho_l \frac{\partial \phi_l}{\partial x_j}}\right)}_{\substack{\textrm{diffusive}\\ \textrm{flux}}} - \underbrace{\left(\reallywidehat{\rho_l\phi_l(1 - \phi_l) n_{lj}}\right)}_{\substack{\textrm{sharpening}\\ \textrm{flux}}} \right\}\right\vert_{(m\pm\frac{1}{2})}.
\end{equation*}
The sharpening flux is a product of four variables $\rho_l$, $\phi_l$, $1 - \phi_l$, and $n_{li}$. Now, approximating this term using the quartic-split form defined in Eq. \eqref{eq:quartic_split_form}, we can derive the remaining numerical fluxes, using the consistency conditions in Section \ref{sec:consistency_sum}, as 
\begin{equation}
\begin{aligned}
    &\hat{\Phi}_j\rvert_{(m\pm\frac{1}{2})} =  \xbar{\phi}^{(m\pm\frac{1}{2})} \xbar{u}_j^{(m\pm\frac{1}{2})},\\
    &\hat{F}_j\rvert_{(m\pm 1/2)} = \sum_{l=1}^2 \left(\xbar{\rho_l}^{(m\pm\frac{1}{2})} \reallywidehat{a}_{lj} \rvert_{(m\pm\frac{1}{2})} \right),\\
    &\hat{R}_{ij}\rvert_{(m\pm 1/2)} = \left\{ \sum_{l=1}^2 \left(\xbar{\rho_l}^{(m\pm\frac{1}{2})} \reallywidehat{a}_{lj} \rvert_{(m\pm\frac{1}{2})} \right) \right\} \xbar{u}_i^{(m\pm\frac{1}{2})},\\
    &\hat{T}_j\rvert_{(m\pm\frac{1}{2})} = \left\{ \sum_{l=1}^2 \left(\xbar{\rho_l}^{(m\pm\frac{1}{2})} \reallywidehat{a}_{lj} \rvert_{(m\pm\frac{1}{2})} \right) \right\} \frac{u_i\rvert_{(m\pm1)} u_i\rvert_{m}}{2},\\
    &\hat{H}_j\rvert_{(m\pm 1/2)} = \sum_{l=1}^2 \left\{\left(\xbar{\rho_l h_l}\right)^{(m\pm\frac{1}{2})} \reallywidehat{a}_{lj} \rvert_{(m\pm\frac{1}{2})} \right\},\\
\end{aligned}
\label{eq:two_phase_keep_scheme}
\end{equation}
where
\begin{equation*}
    \hat{a}_{lj}\rvert_{(m\pm\frac{1}{2})} = \Gamma \left\{ \epsilon \reallywidehat{ \left( \frac{\partial \phi_l}{\partial x_j}\right)}\rvert_{(m\pm\frac{1}{2})} - \xbar{\phi}_l^{(m\pm\frac{1}{2})}(1 - \xbar{\phi}_l^{(m\pm\frac{1}{2})}) \xbar{n}_{lj}^{(m\pm\frac{1}{2})} \right\}.
\end{equation*}
The numerical fluxes in Eq. \eqref{eq:two_phase_keep_scheme}, in addition to Eq. \eqref{eq:single_phase_keep_scheme}, provides a consistent set of discretization for a two-phase system.  

Here, a quadratic-split form is used for $\hat{C}_j$ as opposed to the divergence form because this results in reduced aliasing errors (see Appendix C).  
In Eq. \eqref{eq:single_phase_keep_scheme}, the numerical fluxes $\hat{M}_{ij}$ and $\hat{I}_j$ assume cubic-split and quadratic-split forms, respectively; and in Eq. \eqref{eq:two_phase_keep_scheme}, $\hat{R}_{ij}$ and $\hat{H}_j$ assume quintic-split and quartic-split forms, respectively. In the context of single-phase flows, alternate discrete forms have been proposed for $\hat{I}_j$ in the literature that do not satisfy the consistency conditions proposed in this work. \citet{chandrashekar2013kinetic} proposed
\begin{equation}
    \hat{I}_j\rvert_{(m\pm\frac{1}{2})} =  \xbar{\rho}^{(m\pm\frac{1}{2})}
    \left(\frac{2 e\rvert_{(m\pm1)} e\rvert_{m}}{e \rvert_{(m\pm1)} + e\rvert_{m}}\right)
    \xbar{u}_j^{(m\pm\frac{1}{2})},
    \label{eq:ie_cons_chand}
\end{equation}
as the form of $\hat{I}_j$ in their approximate entropy conserving flux formulation, whereas, \citet{kennedy2008reduced}, and \citet{kuya2018kinetic} used 
\begin{equation}
    \hat{I}_j\rvert_{(m\pm\frac{1}{2})} =  \xbar{\rho}^{(m\pm\frac{1}{2})} \xbar{e}^{(m\pm\frac{1}{2})} \xbar{u}_j^{(m\pm\frac{1}{2})}.
    \label{eq:ie_cons_kuya}
\end{equation}
We, here, use the form
\begin{equation}
    \hat{I}_j\rvert_{(m\pm\frac{1}{2})} =  (\xbar{\rho e})^{(m\pm\frac{1}{2})} \xbar{u}_j^{(m\pm\frac{1}{2})},
    \label{eq:ie_cons_jain}
\end{equation}
because this is required to satisfy IEC (consistency condition 3).
We denote the $\hat{I}_j$ proposed here in Eq. \eqref{eq:ie_cons_jain} as $QS$, Eq. \eqref{eq:ie_cons_chand} as $CS-H$, and Eq. \eqref{eq:ie_cons_kuya} as $CS$, where $QS$ and $CS$ stands for cubic-split and quadratic-split forms, respectively; and $H$ is used to denote that a harmonic average is used instead of an arithmetic average in Eq. \eqref{eq:ie_cons_chand}. The different split-numerical flux forms of $\hat{I}_j$ are summarized in Table \ref{tab:ie_flux}. Similarly, for $\hat{H}_j$, other splitting forms can be derived, but the form in Eq. \eqref{eq:two_phase_keep_scheme} is used here because this is required to satisfy IEC (consistency condition 4). 
\begin{table}[]
\centering
\begin{tabular}{@{}|l|l|l|l|@{}}
\toprule
 &
  CS-H &
  CS &
  QS \\ \midrule
$\hat{I}_j\rvert_{(m\pm\frac{1}{2})}$ &
  $\xbar{\rho}^{(m\pm\frac{1}{2})}\left(\frac{2 e\rvert_{(m\pm1)} e\rvert_{m}}{e \rvert_{(m\pm1)} + e\rvert_{m}}\right)\xbar{u}_j^{(m\pm\frac{1}{2})}$ &
  $\xbar{\rho}^{(m\pm\frac{1}{2})} \xbar{e}^{(m\pm\frac{1}{2})} \xbar{u}_j^{(m\pm\frac{1}{2})}$ &
  $(\xbar{\rho e})^{(m\pm\frac{1}{2})} \xbar{u}_j^{(m\pm\frac{1}{2})}$ \\ \bottomrule
\end{tabular}
\caption{Split-numerical flux forms for $\hat{I}_j$.}
\label{tab:ie_flux}
\end{table}

We showed that all three $QS$, $CS$, and $CS-H$ forms of $\hat{I}_j$ result in approximate conservation of entropy for single-phase flows, when used consistently with the other fluxes in Eqs. \eqref{eq:single_phase_keep_scheme} \citep{jain2020keep}. However, only the proposed $QS$ form of $\hat{I}_j$ results in conservation of entropy for two-phase flows. The $CS$ and $CS-H$ forms do not satisfy the consistency conditions proposed in Section \ref{sec:consistency_sum} and therefore result in spurious mixing of internal energy at the interface. This is illustrated using the numerical simulations in Section \ref{sec:under_resolved}. This alludes to the fact that the set of consistent numerical fluxes that can be derived for a single-phase Navier-Stokes system that results in exact conservation of kinetic energy and approximate conservation of entropy is not unique.    

\section{Coarse-grid simulations \label{sec:under_resolved}}


It is well known that the coarse-grid simulation of high-$Re$ turbulent flow that does not resolve turbulence is a good test of the robustness of a numerical method. 
A formulation that does not conserve kinetic energy and entropy would not be stable because of the lack of grid resolution to support dissipative scales. These types of under-resolved simulations at high-$Re$ were previously studied, for single-phase flows, by \citet{honein2004higher}, \citet{honein2005}, \citet{mahesh2004numerical}, \citet{hou2005robust}, \citet{subbareddy2009fully}, \citet{gassner2016split}, \citet{kuya2018kinetic}, and \citet{coppola2019numerically}. Here, we propose that the under-resolved simulations at high-$Re$ can also be used as a test of robustness of a numerical method for two-phase flows. 
In this section, we present coarse-grid single-phase and two-phase simulations, at infinite $Re$, of Taylor-Green vortex and isotropic turbulence to evaluate the numerical stability of the proposed KEEP scheme.
Throughout this work, an RK4 time-stepping scheme is used for time discretization, and the temporal discretization errors are assumed to be small. {The time-step size is chosen such that the acoustic time scales of both phases are resolved as $\Delta t = CFL \min{(\Delta x/c_1,\Delta x/c_2)}$, where $CFL$ is the Courant–Friedrichs–Lewy number chosen to be about $0.65$ in this work.}
{Throughout this work, the interface-regularization parameters are chosen to be $\Gamma=|\vec{u}|_{max}$ and $\epsilon=\Delta x$, where $|\cdot|_{max}$ represents the maximum value in the domain.}

To evaluate the stability, we look at the time evolution of the total kinetic energy, $K$, of the system, normalized by the initial kinetic energy $K_o$, defined as
\begin{equation}
\frac{K}{K_o} = \frac{\int_{\Omega} (\rho u_i u_i/2) dV}{\int_{\Omega} (\rho u_i u_i/2)_o dV} = \frac{\sum_m (\rho u_i u_i)\rvert_m}{\sum_m (\rho u_i u_i)_o\rvert_m},
\end{equation}
and the relative change in the total entropy of each of the phases $l$ defined as
\begin{equation}
 \frac{\Delta S_l}{S_{lo}} = \frac{\int_{\Omega} \left\{ (\rho_l \phi_l s_l) - (\rho_l \phi_l s_l)_o \right\}dV }{\lvert\int_{\Omega} (\rho_l \phi_l s_l)_o dV\rvert} = \frac{\sum_m \left\{ \left.(\rho_l \phi_l s_l)\right\rvert_m - \left.(\rho_l \phi_l s_l)_o\right\rvert_m \right\} }{\lvert\sum_m \left.(\rho_l \phi_l s_l)_o\right\rvert_m\rvert },     
\end{equation}
respectively, where the subscript $o$ denotes that the quantity is evaluated at the initial time of $t=0$; $dV=\Pi_{j=1}^3 \Delta x_j$ is the cell volume; and $s_l$, for a stiffened-gas EOS, can be defined as
\begin{equation}
 s_l = c_{vl} \ln{\left(\frac{p + \pi_l}{\rho_l^{\gamma_l}}\right)} + s_{lo},     
\end{equation}
where $c_{vl}$ is the specific heat at constant volume for phase $l$.

\subsection{Inviscid Taylor-Green vortex \label{sec:tg}}

Here, an under-resolved simulation of Taylor-Green vortex at infinite $Re$ is presented. The Taylor-Green vortex is a useful test case to evaluate the proposed KEEP scheme because it possesses important physical properties of a turbulent flow, albeit requiring a simple initial setup.  
\begin{figure}
    \centering
    \includegraphics[width=0.4\textwidth]{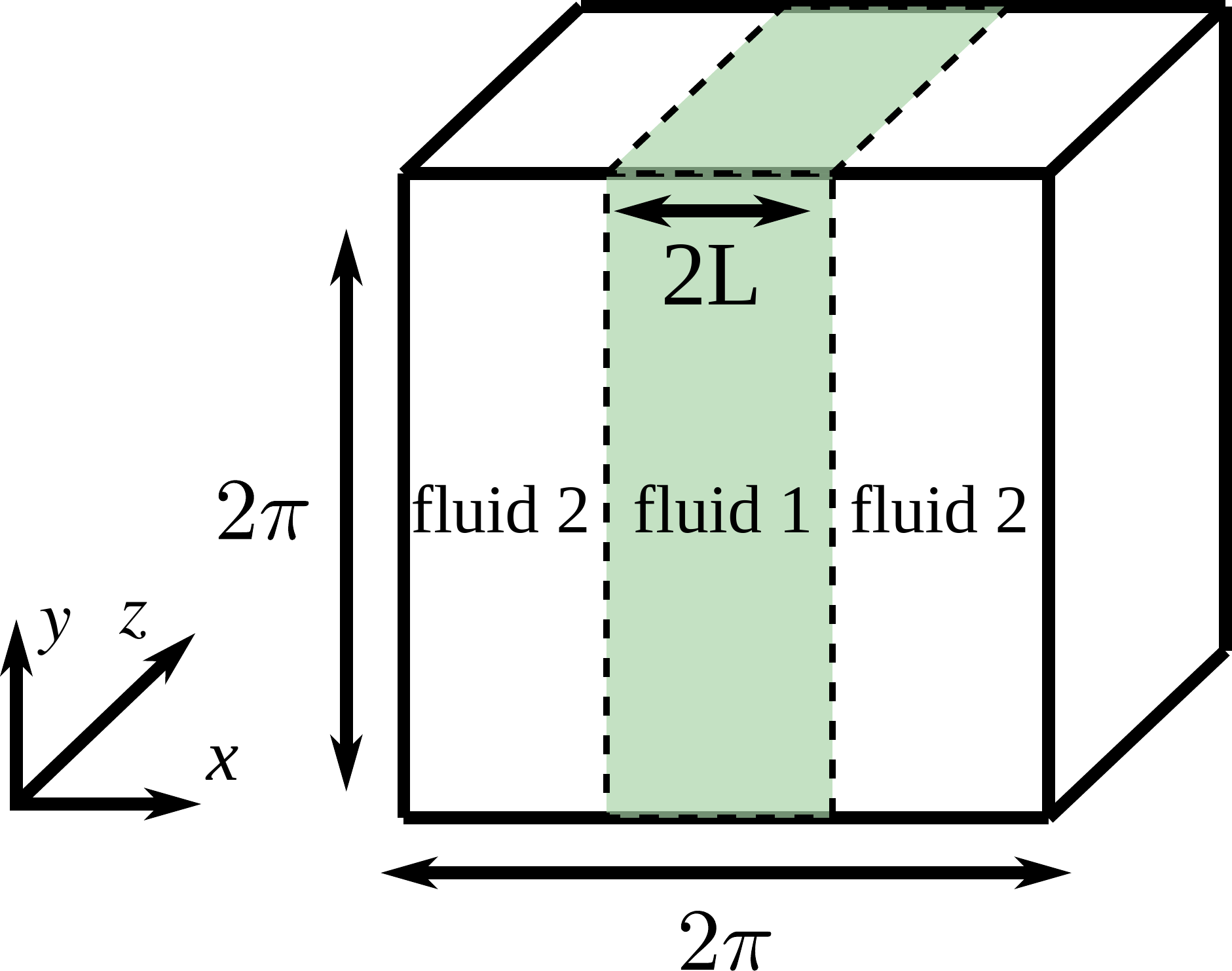}
    \caption{A schematic of the domain used in this work for the coarse-grid simulations of Taylor-Green vortex and isotropic turbulence simulations.}
    \label{fig:tgv_domain}
\end{figure}


The setup consists of a three-dimensional triply periodic domain of size $2\pi\times2\pi\times2\pi$, discretized into a uniform grid of size $64\times64\times64$. No subgrid-scale model is used in the computation. A slab of size $2L$ of fluid $1$ is introduced in to the domain surrounded by fluid $2$, as shown in Figure \ref{fig:tgv_domain}. The initial densities of fluids $1$ and $2$ are chosen to be $\rho_{1o}$ and $\rho_{2o}$, respectively. The viscosities of fluids $1$ and $2$ are zero: $\mu_1=\mu_2=0$. The surface tension between the two fluids is set to zero. The material properties of the fluids $1$ and $2$ are chosen to be: $\gamma_1 = \gamma_2 = \gamma = 1.4$, and $\pi_1=\pi_2=0$. The length scale $L$ is chosen to be $1$. The initial conditions for the volume fraction, density, velocities, and pressure are given by
\begin{equation}
\begin{gathered}
\phi = 1 - \frac{1}{2}\left[1 + \tanh\left\{\frac{\sqrt{(x-\pi)^2}-L}{2 \epsilon}\right\}\right],\\
\rho = \rho_{1o}\phi + \rho_{2o}(1 - \phi),\\
u = M_o\sin\left(\frac{x}{L}\right)\cos\left(\frac{y}{L}\right)\cos\left(\frac{z}{L}\right),\\    
v = -M_o\cos\left(\frac{x}{L}\right)\sin\left(\frac{y}{L}\right)\cos\left(\frac{z}{L}\right),\\
w = 0,\\
p = \frac{1}{\gamma} + \frac{\rho M_o^2}{16}\left\{ \cos\left(\frac{2x}{L}\right) + \cos\left(\frac{2y}{L}\right) \right\} \left\{ \cos\left(\frac{2z}{L}\right) + 2 \right\},
\end{gathered}
\end{equation}
respectively, where $x$, $y$, and $z$ are the spatial coordinates; $u$, $v$, and $w$ are the velocity components along the $x$, $y$, and $z$ directions, respectively; and, $M_o$ is the initial Mach number of the flow defined based on the surrounding fluid (fluid $2$) properties. The $Re$ of this case is, therefore, $Re=\rho UL/\mu=\infty$, where $U$ is the characteristic velocity. Note that when the densities of the two fluids are equal, the setup reduces to a single-phase system (or a pseudo-two-phase system).



\subsubsection{Single-phase flow}


Here, the initial densities of the two fluids are chosen to be $\rho_{1o}=1$ and $\rho_{2o}=1$. Since there is no density difference across the interface between fluids $1$ and $2$ and all other material properties of the two fluids are either zero or equal, the system of equations discretely reduce to the single-phase Navier-Stokes system in Section \ref{sec:reduced}. Using the KEEP scheme proposed in this work in Eq. \eqref{eq:single_phase_keep_scheme}, we compare the three numerical flux forms for $\hat{I}_j$, namely, the $CS, CS$-$H,$ and $QS$ forms, introduced in Section \ref{sec:keep}, and keeping all other fluxes unchanged. 

The time evolution of the total kinetic energy and the individual entropy of fluids $1$ and $2$ are plotted in Figures \ref{fig:dr1_M0_05} and \ref{fig:dr1_M0_2} for $M_o=0.05$ and $M_o=0.2$, respectively. Note that the properties of fluids $1$ and $2$ are the same and hence the tags $1$ and $2$ represent the same fluid, but that are in separate regions in the domain. All three methods conserve kinetic energy and entropy of the system and therefore are stable for single-phase flows. This further verifies that the consistent set of numerical fluxes that result in conservation of kinetic energy and entropy is not unique for a single-phase Navier-Stokes system.
\begin{figure}
    \centering
    \includegraphics[width=\textwidth]{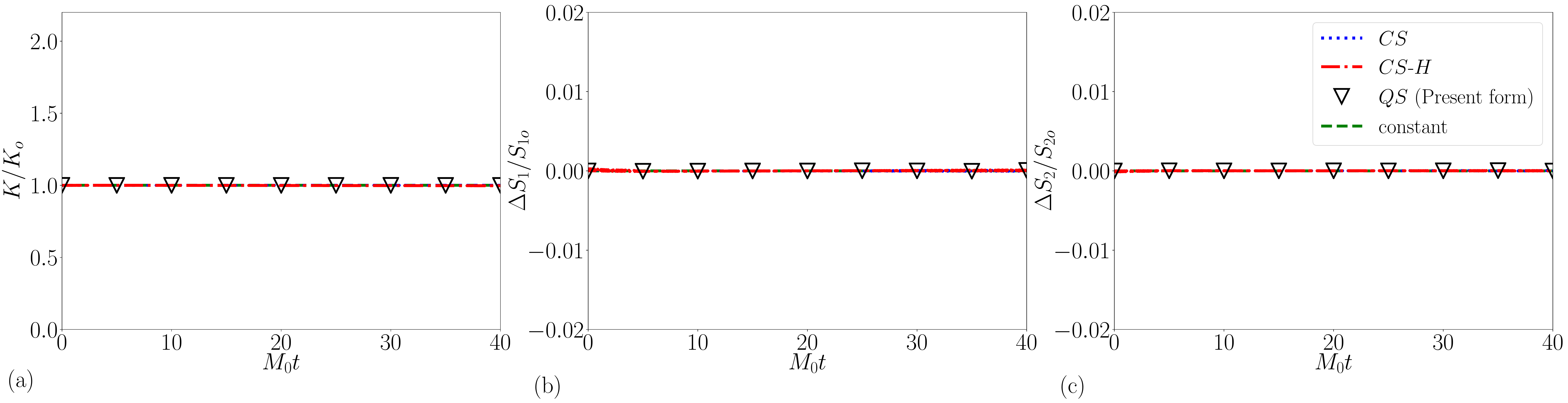}
    \caption{Inviscid single-phase Taylor-Green case ($\rho_2/\rho_1=1$) with $M_o=0.05$. (a) The evolution of the total kinetic energy in the domain $K$ normalized by the initial value of the kinetic energy in the domain $K_o$. (b) The change in entropy of fluid $1$ with time, normalized by the initial value. (c) The change in entropy of fluid $2$ with time, normalized by the initial value. Note that, here, fluids $1$ and $2$ are the same and hence the tags $1$ and $2$ represent the same fluid, but that are in separate regions in the domain.}
    \label{fig:dr1_M0_05}
\end{figure}
\begin{figure}
    \centering
    \includegraphics[width=\textwidth]{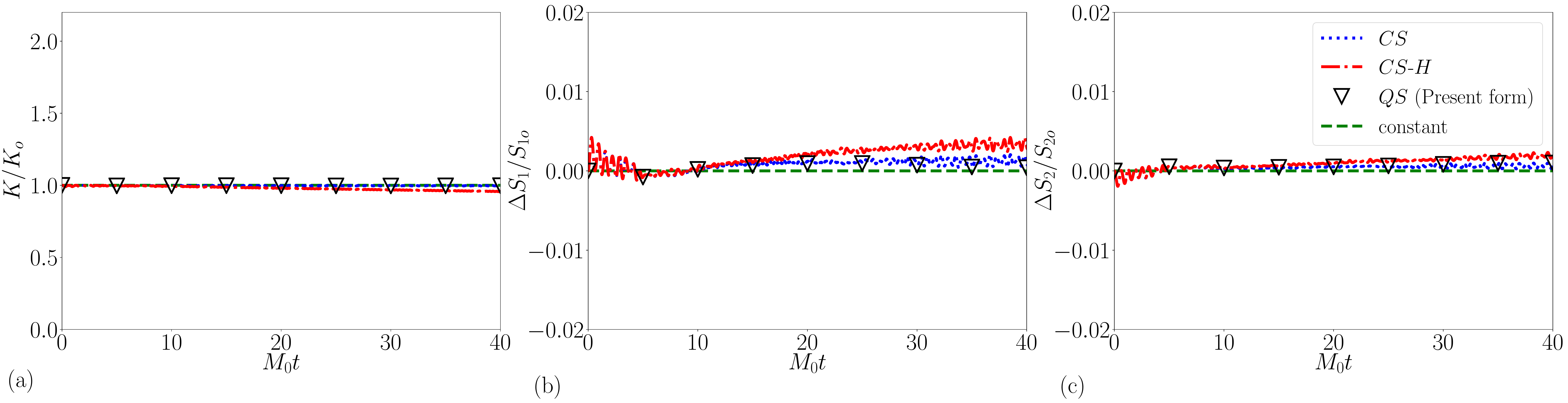}
    \caption{Inviscid single-phase Taylor-Green case ($\rho_2/\rho_1=1$) with $M_o=0.2$. (a) The evolution of the total kinetic energy in the domain $K$ normalized by the initial value of the kinetic energy in the domain $K_o$. (b) The change in entropy of fluid $1$ with time, normalized by the initial value. (c) The change in entropy of fluid $2$ with time, normalized by the initial value. Note that, here, fluids $1$ and $2$ are the same and hence the tags $1$ and $2$ represent the same fluid, but that are in separate regions in the domain.}
    \label{fig:dr1_M0_2}
\end{figure}

\subsubsection{Two-phase flow \label{sec:nonunity_tg}}


In this section, the initial densities of the two fluids are chosen to be $\rho_{1o}=0.1$ and $\rho_{2o}=1$ and the setup is, therefore, a two-phase system with non-unity density difference across the interface between fluids $1$ and $2$. Now using the KEEP scheme proposed in this work in Eqs. \eqref{eq:single_phase_keep_scheme}-\eqref{eq:two_phase_keep_scheme}, we again compare the three numerical flux forms for $\hat{I}_j$, keeping all other fluxes unchanged. Snapshots of the flow are shown in Figure \ref{fig:tgv} at various times, illustrating the breakup of the fluid slab into bubbles due to the breakdown of the underlying Taylor-Green vortex.

The time evolution of the total kinetic energy and the individual entropy of fluids $1$ and $2$ are plotted in Figures \ref{fig:dr10_M0_05} and \ref{fig:dr10_M0_2} for $M_o=0.05$ and $M_o=0.2$, respectively. Unlike for the single-phase system, here only the $QC$ form proposed in this work conserves kinetic energy and entropy, and is therefore stable for two-phase flows. The $CS$ and $CS$-$H$ forms on the other hand, do not satisfy the consistency conditions proposed in Section \ref{sec:consistency_sum}, and therefore, do not conserve kinetic energy and entropy. The simulations with these forms diverge. Note that, here, different numerical flux forms for $\hat{I}_j$ are compared because these forms are used before in the literature in the context of single-phase flows. Similarly, a cubic-split form for $\hat{H}_j$ can be constructed instead of the quadratic-split form proposed in Eq. \eqref{eq:two_phase_keep_scheme} and can be shown that the resulting simulation would not be stable since the cubic-split form for $\hat{H}_j$ does not satisfy the consistency conditions. Appendix D discusses other possible inconsistent flux formulations and the effect of these formulations on the numerical stability. 
\begin{figure}
    \centering
    \includegraphics[width=\textwidth]{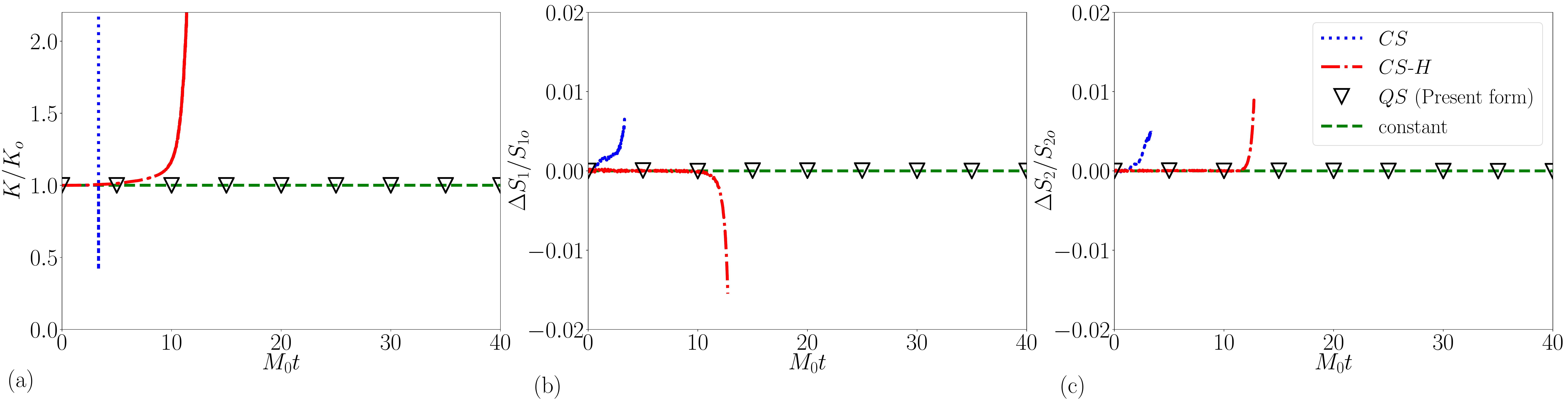}
    \caption{Inviscid two-phase Taylor-Green case with $\rho_2/\rho_1=10$ and $M_o=0.05$. (a) The evolution of the total kinetic energy in the domain $K$ with time normalized by the initial value of the kinetic energy in the domain $K_o$. (b) The change in entropy of fluid $1$ with time, normalized by the initial value. (c) The change in entropy of fluid $2$ with time, normalized by the initial value.}
    \label{fig:dr10_M0_05}
\end{figure}
\begin{figure}
    \centering
    \includegraphics[width=\textwidth]{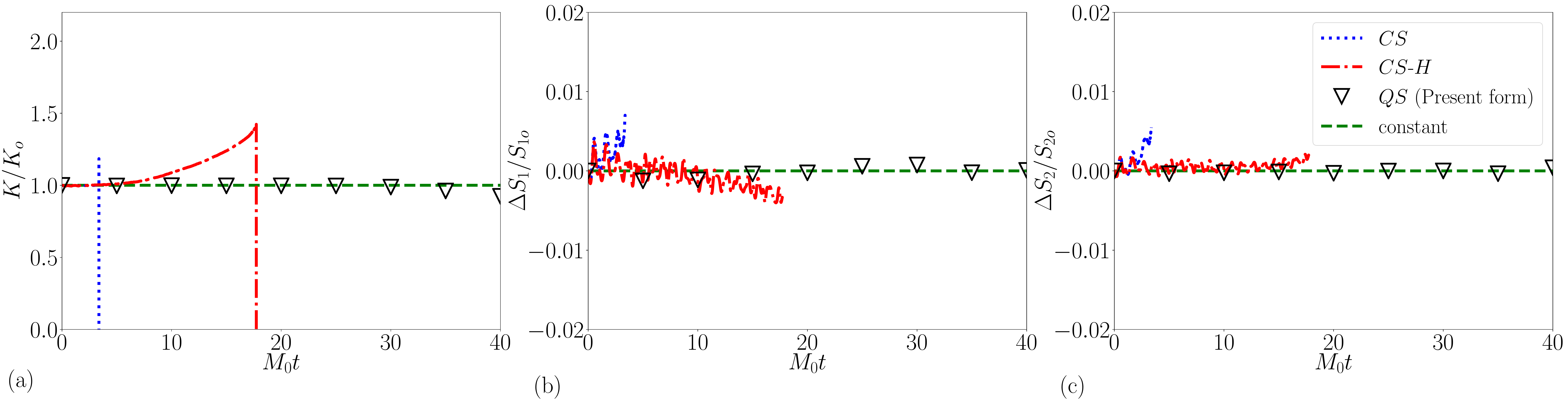}
    \caption{Inviscid two-phase Taylor-Green case with $\rho_2/\rho_1=10$ and $M_o=0.2$. (a) The evolution of the total kinetic energy in the domain $K$ with time normalized by the initial value of the kinetic energy in the domain $K_o$. (b) The change in entropy of fluid $1$ with time, normalized by the initial value. (c) The change in entropy of fluid $2$ with time, normalized by the initial value.}
    \label{fig:dr10_M0_2}
\end{figure}
\begin{figure}
    \centering
    \includegraphics[width=0.8\textwidth]{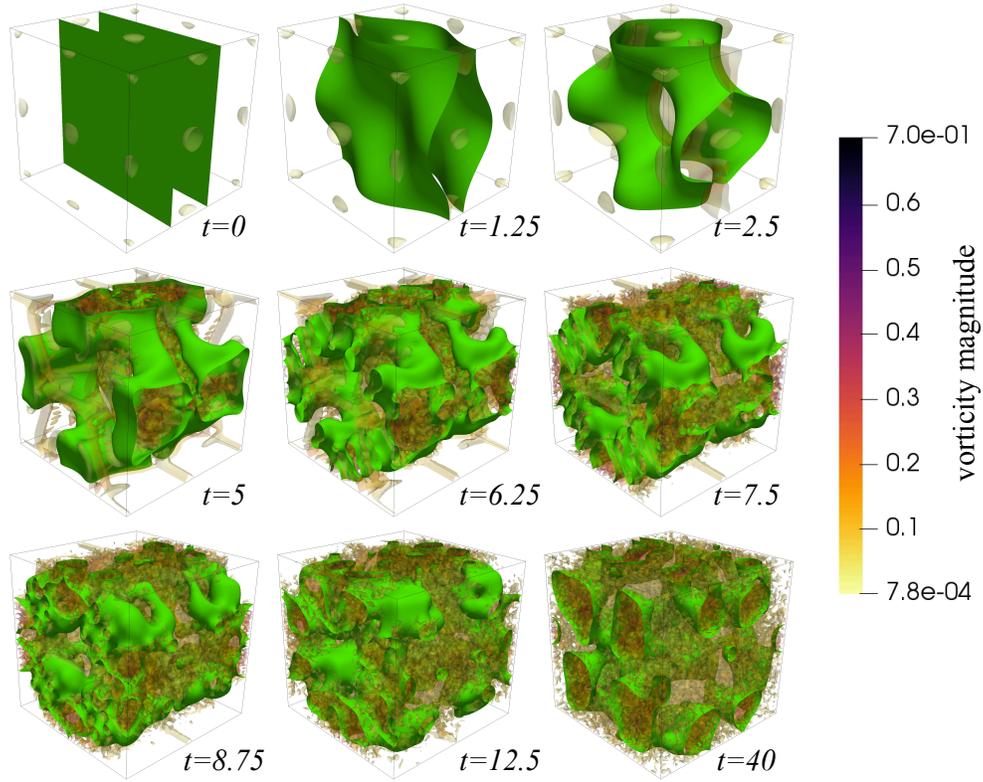}
    \caption{Snapshots of the inviscid two-phase Taylor-Green vortex at various times with the proposed KEEP scheme in this work. The solid surface represents the interface ($\phi=0.5$ contour) between fluid $1$ and fluid $2$. The translucent surface represents an isosurface of the Q-criterion colored by the vorticity magnitude.}
    \label{fig:tgv}
\end{figure}

Finally, the proposed KEEP scheme is verified for different initial Mach numbers and densities of the fluids in Figure \ref{fig:dr1000_M0}. As $M_o$ increases, the kinetic energy does not stay constant due to the exchange of energy between the kinetic energy and internal energy through reversible pressure work. 
\begin{figure}
    \centering
    \includegraphics[width=\textwidth]{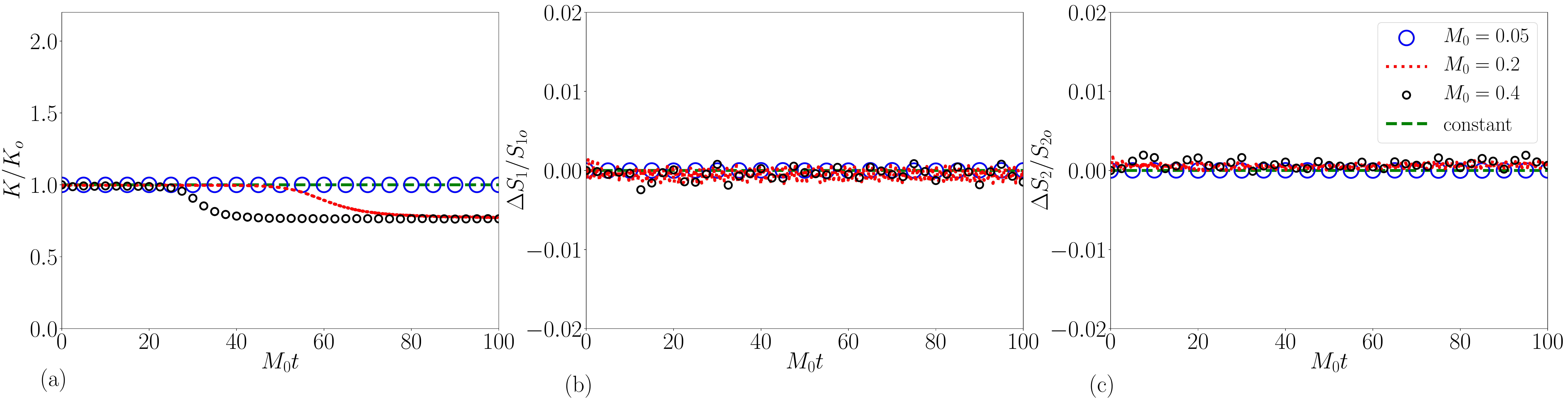}
    \caption{Inviscid two-phase Taylor-Green case with $\rho_2/\rho_1=1000$ for different Mach numbers (a) The evolution of the total kinetic energy in the domain $K$ with time normalized by the initial value of the kinetic energy in the domain $K_o$. (b) The change in entropy of fluid $1$ with time, normalized by the initial value. (c) The change in entropy of fluid $2$ with time, normalized by the initial value.}
    \label{fig:dr1000_M0}
\end{figure}

\subsection{Two-phase isotropic turbulence}

In this section, an under-resolved simulation of two-phase isotropic turbulence at infinite $Re$ is presented. The setup consists of a three-dimensional triply periodic domain of size $2\pi\times2\pi\times2\pi$, discretized into a uniform grid of size $32\times32\times32$. No subgrid-scale model is used in the computation. A slab of size $2L$ of fluid $1$ is introduced in to the domain surrounded by fluid $2$, as shown in Figure \ref{fig:tgv_domain}. The initial densities of fluids $1$ and $2$ are chosen to be $\rho_{1o}=0.1$ and $\rho_{2o}=1$, respectively. The viscosities of fluids $1$ and $2$ are zero: $\mu_1=\mu_2=0$. The material properties of the fluids $1$ and $2$ are chosen to be: $\gamma_1 = \gamma_2 = \gamma = 1.4$, and $\pi_1=\pi_2=0$. The surface tension between the two fluids is set to zero. The length scale $L$ is chosen to be $1$. The initial conditions for the volume fraction and density are the same as the ones in Section \ref{sec:tg}. An initial energy spectrum used by \citet{honein2005} defines the initial velocity field.
\begin{figure}
    \centering
    \includegraphics[width=\textwidth]{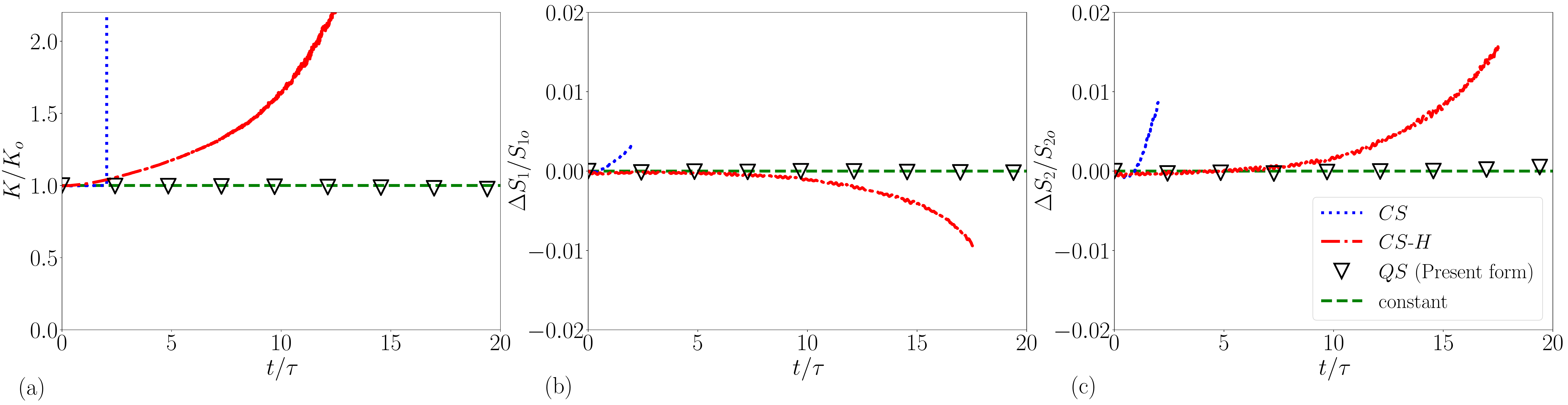}
    \caption{Inviscid two-phase isotropic turbulence case with $\rho_2/\rho_1=10$ and $M_{to}=0.07$. (a) The evolution of the total kinetic energy in the domain $K$ with time normalized by the initial value of the kinetic energy in the domain $K_o$. (b) The change in entropy of fluid $1$ with time, normalized by the initial value. (c) The change in entropy of fluid $2$ with time, normalized by the initial value.}
    \label{fig:dr10_M0_07}
\end{figure}

Using the KEEP scheme proposed in this work in Eqs. \eqref{eq:single_phase_keep_scheme}-\eqref{eq:two_phase_keep_scheme}, we again compare the three numerical flux forms for $\hat{I}_j$, keeping all other fluxes unchanged. The time evolution of the total kinetic energy and the individual entropy of fluids $1$ and $2$ are plotted in Figure \ref{fig:dr10_M0_07} for an initial turbulent Mach number of $M_{to}=0.07$. Similar to the two-phase Taylor-Green system in Section \ref{sec:nonunity_tg}, only the $QC$ form proposed in this work conserves kinetic energy and entropy, and is therefore stable. The $CS$ and $CS$-$H$ forms, however, do not satisfy the consistency conditions and therefore do not conserve kinetic energy and entropy. The simulations with these forms diverge.  


\section{Higher-resolution simulations of droplet-laden isotropic turbulence \label{sec:resolved}}



In this section, simulations of droplet-laden isotropic turbulence is presented to evaluate the accuracy of the method at finite $Re$. The setup consists of a three-dimensional triply periodic domain of size $\pi/2\times\pi/2\times\pi/2$, and is discretized into a uniform grid of size $256^3$. No subgrid-scale model is used in the computation. A 1000, initially spherical, droplets of fluid $1$ is introduced to the domain surrounded by fluid $2$. The initial Taylor-Reynolds number of the flow is $Re_{\lambda}\approx75$ and the initial turbulent Mach number is $M_{to}=0.07$, based on the surrounding fluid properties. The initial diameter of the seeded droplets is $d_o\approx2\lambda$. Here, $Re_{\lambda}$ was chosen to match the incompressible droplet-laden isotropic turbulence study by \citet{dodd2016interaction}. 

The initial density of fluid $2$ is $\rho_{2o}=1$ and the density of the fluid $1$ is varied: $\rho_{1o}=1,10,100,1000$. A reference case of single-phase homogeneous isotropic turbulence (HIT) is also simulated with only fluid $2$ in the domain. The kinematic viscosities of fluids $1$ and $2$ are assumed to be the same. The material properties of the fluids $1$ and $2$ are chosen to be: $\gamma_1 = \gamma_2 = \gamma = 1.4$, and $\pi_1, \pi_2$ are chosen such that the values of $M_{to}$ are identical in two fluids. 
The surface tension, $\sigma$, between the two fluids is chosen such that the initial turbulent Weber number is $We_{rms}=\rho_{2o} U_{rms}^2 d_o/\sigma=1$. An initial energy spectrum used by \citet{honein2005} defines the initial velocity field. 

Snapshots from the simulations of droplet-laden isotropic turbulence for all four values of $\rho_1/\rho_2$, at $t=20$, is presented in Figure \ref{fig:resolved_drop_hit}. The only difference between the single-phase case and the two-phase case with $\rho_1/\rho_2=1$ is the presence of the surface tension effects at the droplet interface. This inhibits breakup of the droplet fluid, and hence, the droplet size increases with time. With an increase in the density of the droplet fluid, it is evident that the sizes of the drops are smaller. This is due to the increase in the inertial effects over the surface tension effects that counter the breakup process. 
\begin{figure}
    \centering
    \includegraphics[width=0.75\textwidth]{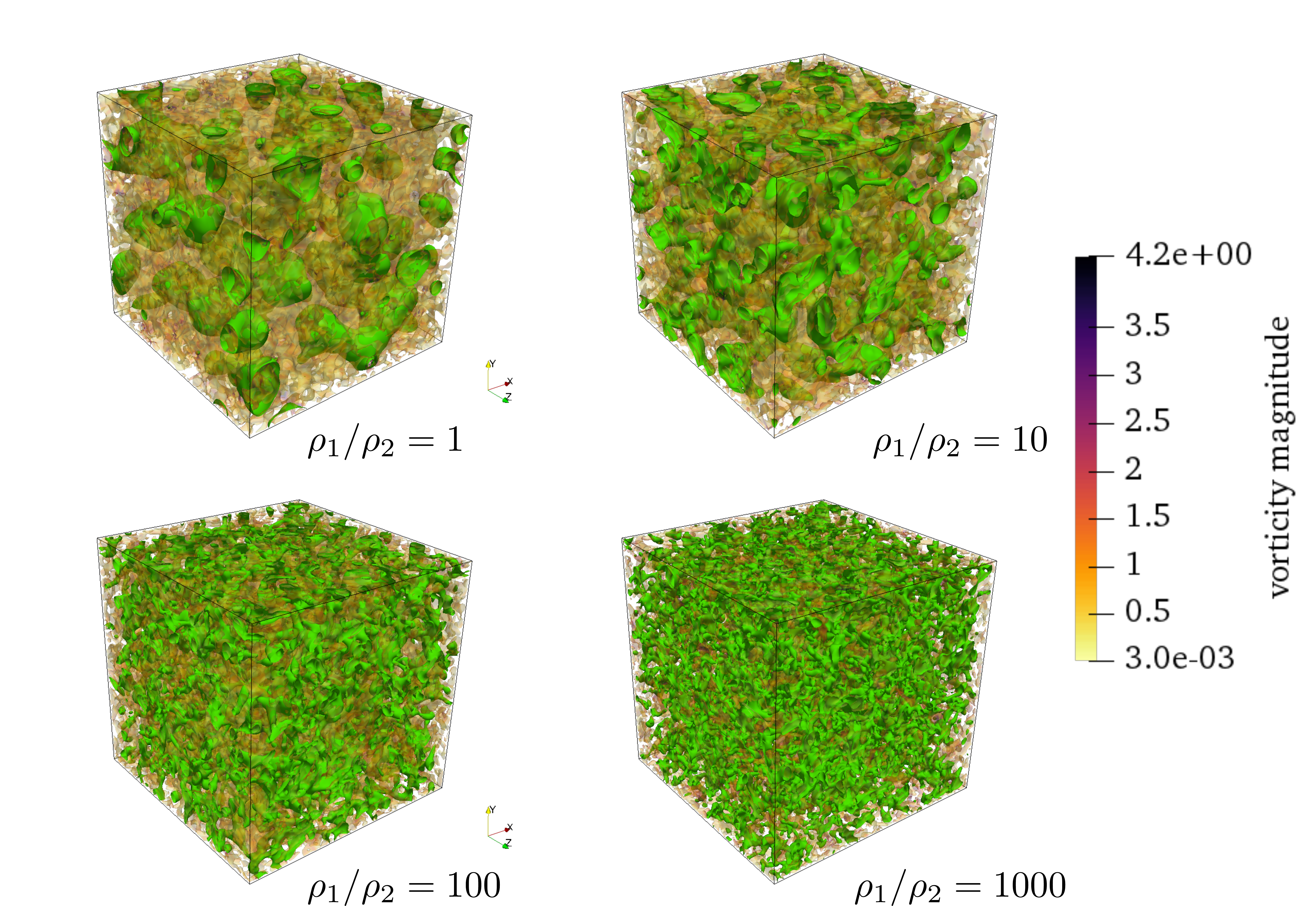}
    \caption{Snapshots of the droplet-laden isotropic turbulence simulations for various values of $\rho_1/\rho_2$ at $t=20$. The solid surface represents the interface ($\phi=0.5$ contour) between fluid $1$ and fluid $2$. The translucent surface represents an isosurface of the Q-criterion colored by the vorticity magnitude.}
    \label{fig:resolved_drop_hit}
\end{figure}

The time evolution of the total kinetic energy is plotted in Figure \ref{fig:resolved_ke} for all four values of $\rho_1/\rho_2$ along with the reference single-phase case. The relative enhancement of the decay rate of total kinetic energy at the early times, due to the presence of drops (see the differences in $K$ between the two-phase case of $\rho_1/\rho_2=1$ and the single-phase case), and with an increase in the density of the droplet fluid, is in very good agreement with the observations made by \citet{dodd2016interaction}. However, for the $\rho_1/\rho_2=1$ case, the decay of $K$ slows down compared to the single-phase case at later times and a crossover of $K$ at $t\approx10$ can be seen in Figure \ref{fig:resolved_ke}. This behavior can be explained using the effect of surface tension on the flow. Writing the evolution equation for $K$ as
\begin{equation}
    \frac{d K}{dt} = -E + \Psi
\end{equation}
where $E$ and $\Psi$ are the total dissipation rate and the rate of change of surface energy, which scales as
\begin{equation}
    \Psi \sim -\frac{dA}{dt},
\end{equation}
where $A$ is the total surface area of the drops in the domain. As the droplet sizes increase with time for the $\rho_1/\rho_2=1$ case at later time in the simulation, the total surface area $A$ decreases. Therefore, $\Psi$ is positive and acts as a source of $K$. This results in slower decay of $K$ for the two-phase case with $\rho_1/\rho_2=1$ compared to the single-phase case, and hence, a crossover in $K$ is seen at $t\approx10$. A similar behavior is expected for the case of $\rho_1/\rho_2=10$, but at a much later time, due to a relative slower increase in the droplet sizes compared to the case of $\rho_1/\rho_2=1$. Figure \ref{fig:resolved_ke} confirms this behavior, and a crossover in $K$ can be seen for this case at $t\approx 18$.
\begin{figure}
    \centering
    \includegraphics[width=0.6\textwidth]{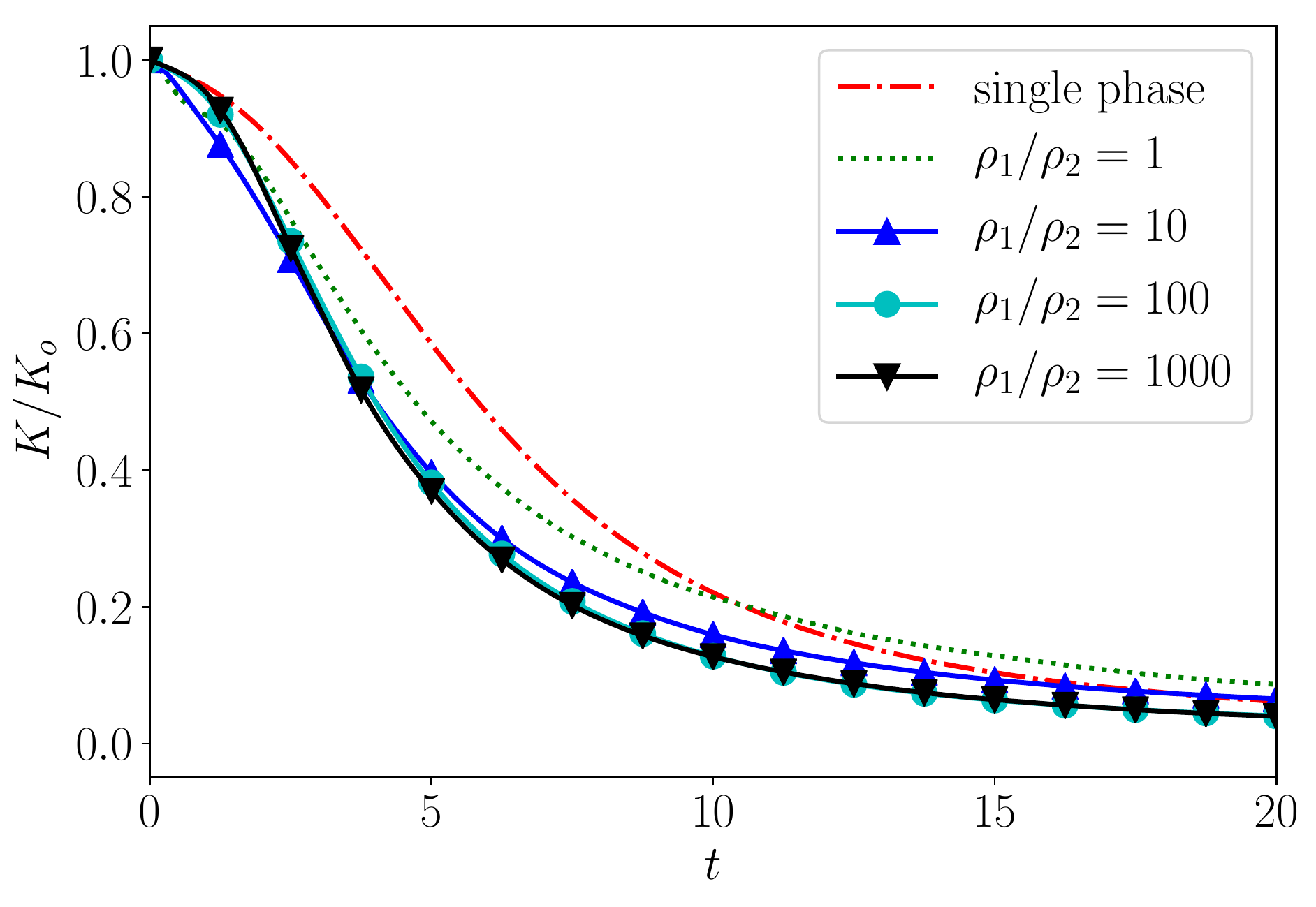}
    \caption{The time evolution of the total kinetic energy in the domain $K$, normalized by the initial value of the kinetic energy in the domain $K_o$, for the droplet-laden isotropic turbulence simulations for various values of $\rho_1/\rho_2$ and a reference single-phase case.}
    \label{fig:resolved_ke}
\end{figure}

\section{Summary and remarks \label{sec:conclusion}}



In this work, we developed a general framework for the derivation of consistent numerical fluxes for compressible single-phase and two-phase flows, and proposed a set of consistency conditions between the numerical fluxes of volume fraction, mass, momentum, kinetic energy, and internal energy. The proposed consistency conditions between the fluxes are required for achieving discrete conservation of global kinetic energy, preservation of local kinetic energy, and for maintaining interface-equilibrium condition. 
For incompressible flows, conservation of global kinetic energy is known to be a sufficient condition to achieve stable numerical simulations. However, for compressible flows, in addition to conserving global kinetic energy, more constraints such as local kinetic energy preservation and interface-equilibrium condition are required to maintain an approximate discrete conservation of entropy, and to achieve stable numerical simulations.
The proposed consistency conditions for the numerical fluxes are general enough that they can be adopted for other multiphysics flow problems, such as multicomponent flows. 

Using the proposed framework, a numerical scheme that satisfies the consistency conditions was derived for both single-phase and two-phase flows and was verified that it results in an exact conservation of the kinetic energy and approximate conservation of entropy (a KEEP scheme) in the absence of pressure work, viscosity, thermal diffusion effects, and time-differencing errors. We also showed that a KEEP scheme is not unique for single-phase flows, and various forms of numerical fluxes can be derived that satisfy the consistency conditions and hence would result in exact conservation of kinetic energy and approximate conservation of entropy.

A conservative diffuse-interface method with a five-equation model was used as the interface-capturing method for modeling the system of compressible two-phase flows in this work. However, the proposed consistency conditions are general, and can be used with any other interface-capturing method to derive a KEEP scheme.

We used coarse-grid simulations of single-phase and two-phase Taylor-Green vortex and isotropic turbulence at infinite $Re$, to test the numerical stability of the proposed scheme. The observations from the numerical experiments verified that the proposed scheme results in conservation of kinetic energy and entropy, and hence, the scheme is superior in terms of maintaining stability for long-time numerical integrations. A higher-resolution simulation of droplet-laden decaying isotropic turbulence case (at finite $Re$) is also presented at the end, and the effect of presence of droplets on the flow is studied. This test case illustrates the stability and accuracy of the proposed method in a complex setting.

\section*{Acknowledgments} 

S. S. J. was supported by the Franklin P. and Caroline M. Johnson Stanford Graduate Fellowship, and the authors acknowledge the support from the Predictive Science Academic Alliance Program (PSAAP) III at Stanford University.
The authors also acknowledge the Argonne Leadership Computing Challenge award 2021-22 which provided access to the Theta supercomputer, which was used for simulations in this work. 
A preliminary version of this work has been published as a technical report in the annual publication of the Center for Turbulence Research \citep{jain2020keep} and is available online\footnote{http://web.stanford.edu/group/ctr/ResBriefs/2020/29\_Jain.pdf}. S. S. J. is thankful for Mr. Kihiro Bando, for reading the manuscript and for his comments that helped improve the manuscript.


\section*{Appendix A: Entropy change associated with interface-regularization process}

To quantify the entropy change associated with the interface-regularization process and to illustrate the conservation of entropy in the limit of equilibrium interface state, a simple numerical test is performed. In this test case, the interface thickness is varied in a stationary setup such that the only terms that are non-zero are the interface-regularization terms, facilitating the quantification of entropy change associated with these terms.

The setup consists of an infinite slab of fluid $1$ of width $2L$, where $L=0.25$, surrounded by fluid $2$ in a two-dimensional periodic square domain, as shown in Figure \ref{fig:entropy_drop_case}(a). The domain has dimensions of $[-0.5,0.5]\times[-0.5,0.5]$, discretized into a uniform grid of size $64\times64$. 
The initial conditions for the volume fraction, density, velocities, and pressure are given by
$\phi = 1 - {1}/{2}\left[1 + \tanh\left\{{(|x|-L)}/{(2 \epsilon_o)}\right\}\right]$, $\rho = \rho_{1o}\phi + \rho_{2o}(1 - \phi)$, $u = 0$, $v=0$, and $p=1$, respectively; $\rho_{1o}=1$ and $\rho_{2o}=1000$ are the initial densities of the phases $1$ (air) and $2$ (water), respectively. The viscosities are taken to be zero: $\mu_1=\mu_2=0$. The surface tension is set to zero.
\begin{figure}
    \centering
    \includegraphics[width=0.75\textwidth]{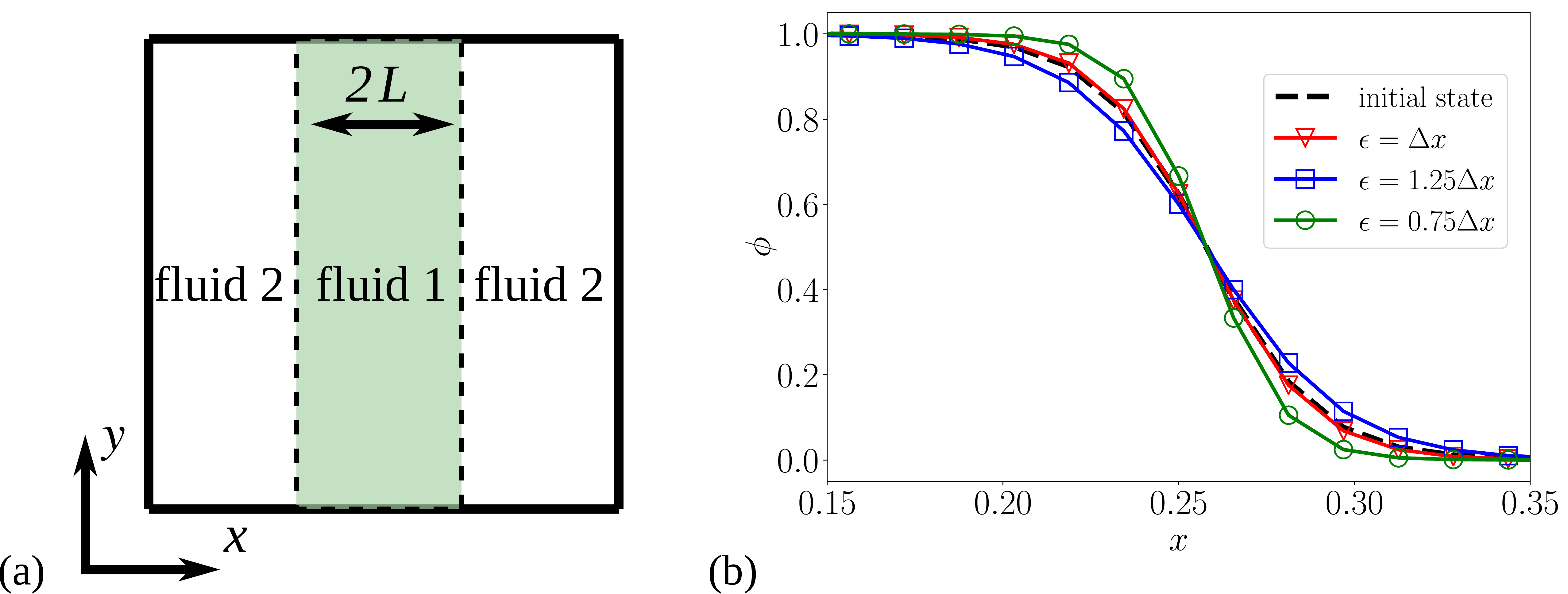}
    \caption{A stationary slab of fluid $1$ in fluid $2$. (a) initial setup, (b) A line plot of $\phi$ along $y=0$ from $x=0$ to $x=0.5$, showing the initial and final states for three different scenarios of $\epsilon$.}
    \label{fig:entropy_drop_case}
\end{figure}
\begin{figure}
    \centering
    \includegraphics[width=\textwidth]{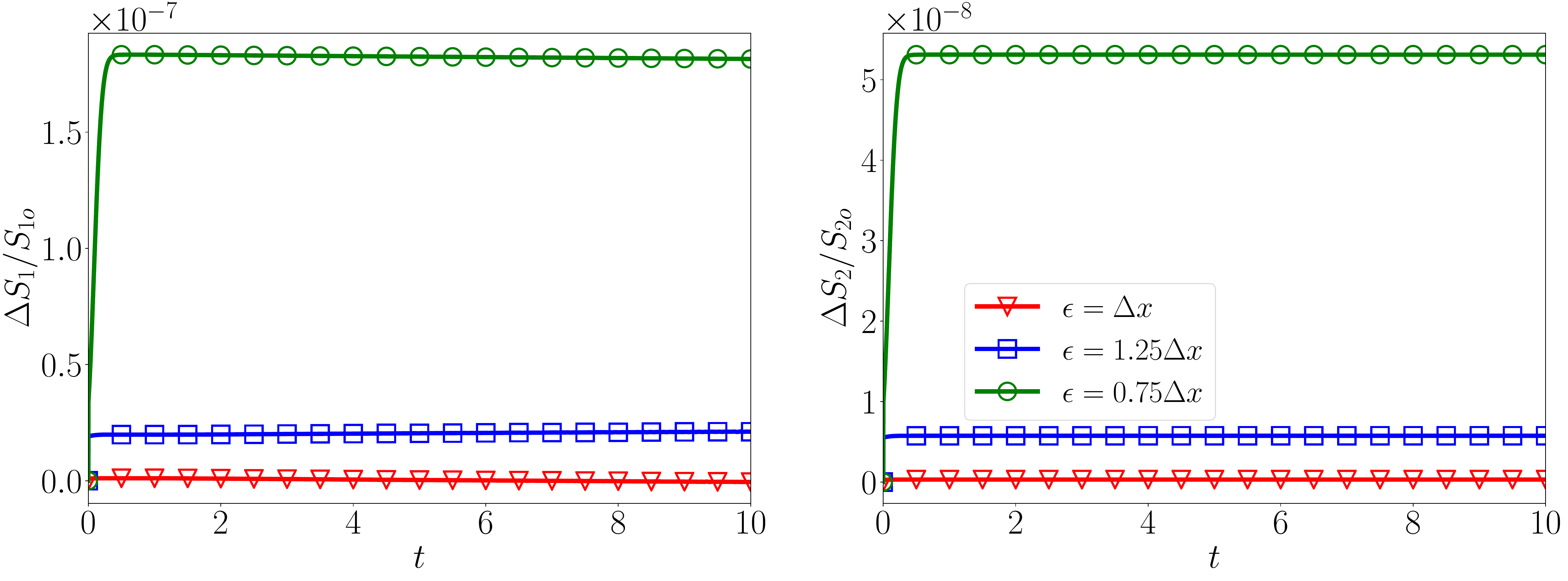}
    \caption{The relative change in total entropy $\Delta S_l/S_{lo}$ of (a) fluid $1$, and (b) fluid $2$, for three different values of $\epsilon$.}
    \label{fig:entropy_drop}
\end{figure}

Three different scenarios are considered with different interface equilibrium states that are controlled using the interface thickness parameter, $\epsilon$: (1) $\epsilon=\Delta x$, (2) $\epsilon=1.25\Delta x$, and (3) $\epsilon=0.75\Delta x$. In all the three cases, volume fraction is initialized using the hyperbolic tangent function with $\epsilon_o=\Delta x$. Hence, the interface is already in an equilibrium state in case $1$. But for cases $2$ and $3$, the interface thickness increases and decreases, respectively, until it reaches a new equilibrium state.
The relative change in total entropy of each phase $l$ is computed as a function of time and are shown in Figures \ref{fig:entropy_drop}(a,b) for all three cases. Since in case $1$, the interface is already in an equilibrium state, the change in entropy of both fluids $1$ and $2$ are negligible. For cases $2$ and $3$, where the interface thickness increases and decreases, respectively, there is a small change (increase) in entropy. But once the interface reaches a new equilibrium state in cases $2$ and $3$, the entropy of both fluids does not change, as expected. The initial and final states of the interface at times, $t=0$ and $t=1$, respectively, are shown in Figure \ref{fig:entropy_drop_case}(b) for all three scenarios. This numerical demonstration verifies that the entropy is conserved in the limit of equilibrium interface state and the change in entropy due to the interface-regularization process is small and positive. Although the change in entropy due to the interface-regularization terms are small and essentially negligible, this does not warrant the use of inconsistent numerical fluxes for these terms. Failing to satisfy the consistency conditions (KEEP scheme) presented in this work, these interface-regularization terms could spuriously contribute to the kinetic energy and entropy of the system (see, Appendix D).

\section*{Appendix B: KEEP scheme satisfies the interface-equilibrium condition \label{sec:proof_iec}}

The IEC provides a consistency condition to check and eliminate the forms of the numerical discretizations that contribute spuriously to the solution. Satisfying the IEC is crucial for a method to conserve kinetic energy and entropy.

\newtheorem{lemma}[theorem]{Lemma}
\begin{lemma}
The proposed KEEP scheme satisfies the IEC defined in Section \ref{sec:ie_consistency}.
\label{lemma:IEC}
\end{lemma}

\begin{proof}
\subsection*{Part (a). Uniform velocity}
Consider a finite-volume discretization of the total mass balance equation in Eq. \eqref{eq:mod_continuity} 
\begin{equation}
\rho\rvert_m^{k+1} - \rho\rvert_m^{k} = -\Delta t \left( \frac{\hat{C}_j\rvert_{(m+\frac{1}{2})} - \hat{C}_j\rvert_{(m-\frac{1}{2})}}{\Delta x_j} \right)^k + \Delta t  \left( \frac{\hat{F}_j\rvert_{(m+\frac{1}{2})} - \hat{F}_j\rvert_{(m-\frac{1}{2})}}{\Delta x_j} \right)^k,
\label{eq:mass_iec}
\end{equation}
where $k$ is the time-step index, and $m$ is the grid index. Now, consider a finite-volume discretization of the momentum equation in Eq. \eqref{eq:momf} and ignoring viscous, surface tension, and gravity terms, and assuming $u_{i}\rvert_m^k=u_0$ and $p\rvert_m^k=p_0$, $\forall m$
\begin{equation}
(\rho u_i)\rvert_m^{k+1} - \rho\rvert_m^{k}u_0 = -\Delta t \left( \frac{\hat{M}_{ij}\rvert_{(m+\frac{1}{2})} - \hat{M}_{ij}\rvert_{(m-\frac{1}{2})}}{\Delta x_j} \right)^k + \Delta t \left( \frac{\hat{R}_{ij}\rvert_{(m+\frac{1}{2})} - \hat{R}_{ij}\rvert_{(m-\frac{1}{2})}}{\Delta x_j} \right)^k.
\end{equation}
Utilizing the consistency conditions in Eqs. \eqref{eq:mom_cons}, \eqref{eq:int_mom_cons}, and assuming $u_{i}\rvert_m^k=u_0$, $\forall m$, the discrete momentum equation can be written as
\begin{equation}
(\rho u_i)\rvert_m^{k+1} - \rho\rvert_m^{k} u_0 = -\Delta t \left( \frac{\hat{C}_j\rvert_{(m+\frac{1}{2})} - \hat{C}_j\rvert_{(m-\frac{1}{2})}}{\Delta x_j} \right)^k u_0 + \Delta t \left( \frac{\hat{F}_j\rvert_{(m+\frac{1}{2})} - \hat{F}_j\rvert_{(m-\frac{1}{2})}}{\Delta x_j} \right)^k u_0.
\end{equation}
Subtracting this from the discrete mass balance equation in Eq. \eqref{eq:mass_iec}, gives $u_i\rvert_{m}^{k+1}=u_0$, $\forall m$.

\subsection*{Part (b). Uniform pressure}
Consider a finite-volume discretization of the total energy equation [Eq. \eqref{eq:energyf}] without the viscous, surface tension, and gravity terms. Subtracting the discrete version of the kinetic energy equation in Eq. \eqref{eq:kineticf} from the discrete total energy equation (note that, this step requires that the consistency conditions in Eqs. \eqref{eq:ke_cons}, \eqref{eq:int_ke_cons} are satisfied), we arrive at
\begin{equation}
\left(\rho e\right)\rvert_m^{k+1} - \left(\rho e \right)\rvert_m^k = -\Delta t \left( \frac{\hat{I}_j\rvert_{(m+\frac{1}{2})} - \hat{I}_j\rvert_{(m-\frac{1}{2})}}{\Delta x_j} \right)^k + \Delta t  \left( \frac{\hat{H}_j\rvert_{(m+\frac{1}{2})} - \hat{H}_j\rvert_{(m-\frac{1}{2})}}{\Delta x_j} \right)^k
\label{eq:ie_iec}
\end{equation}
Now, substituting for $\hat{I}_j\rvert_{(m\pm\frac{1}{2})}$ and $\hat{H}_j\rvert_{(m\pm\frac{1}{2})}$ from Eqs. \eqref{eq:single_phase_keep_scheme}-\eqref{eq:two_phase_keep_scheme}; expressing $h_l$ in terms of $p$ and $\rho_l$ using Eq. \eqref{eq:enthalpy} as $h_l=(p - \beta_l)/\alpha_l + p$; expressing the mixture internal energy in terms of the individual species energies as $\rho e = \sum_{l=1}^2 \rho_l e_l \phi_l$; and assuming $u_i\rvert_{m}^k=u_0$ and $p\rvert_m^k=p_0$, we obtain 
\begin{equation}
\begin{aligned}
\sum_{l=1}^2 \left(\phi_l\rho_le_l\right)\rvert_m^{k+1} - \sum_{l=1}^2 \left(\phi_l\rho_le_l\right)\rvert_m^k = -\Delta t \sum_{l=1}^2\left\{ \frac{\left(\xbar{\rho_le_l\phi_l}^{(m+\frac{1}{2})}\right)-\left(\xbar{\rho_le_l\phi_l}^{(m-\frac{1}{2})}\right)}{\Delta x_j} \right\}^k u_0 \\
+\Delta t \left\{\sum_{l=1}^2 \left(\frac{p_0 - \beta_l}{\alpha_l}\right)\left(\frac{\hat{a}_{lj}\rvert_{(m+\frac{1}{2})} - \hat{a}_{lj}\rvert_{(m-\frac{1}{2})}}{\Delta x_j}\right) \right\}^k,
\end{aligned}
\label{eq:iie_iec}
\end{equation}
and expressing $e_l$ in terms of $p$ using the EOS results in the discretized equation for pressure
\begin{equation}
\begin{aligned}
\left(\sum_{l=1}^2\frac{\phi_{l}\rvert_m}{\alpha_l}\right)^{k+1}p\rvert_m^{k+1} - \left(\sum_{l=1}^2\frac{\phi_{l}\rvert_m\beta_l}{\alpha_l}\right)^{k+1} - \left(\sum_{l=1}^2\frac{\phi_{l}\rvert_m}{\alpha_l} \right)^kp_0 + \left(\sum_{l=1}^2\frac{\phi_{l}\rvert_m\beta_l}{\alpha_l} \right)^k\\
= -\Delta t \left[ \sum_{l=1}^2 \left\{\frac{\left(\frac{\xbar{\phi}_l^{(m+\frac{1}{2})}}{\alpha_l}\right)p_0 - \left(\frac{\xbar{\phi}_l^{(m+\frac{1}{2})}\beta_l}{\alpha_l}\right)-\left(\ \frac{\xbar{\phi}_l^{(m-\frac{1}{2})}}{\alpha_l} \right)p_0 + \left(\frac{ \xbar{\phi}_l^{(m-\frac{1}{2})}\beta_l}{\alpha_l}\right)}{\Delta x_j} \right\}^k \right] u_0\\
+\Delta t \left\{\sum_{l=1}^2 \left(\frac{p_0 - \beta_l}{\alpha_l}\right)\left(\frac{\hat{a}_{lj}\rvert_{(m+\frac{1}{2})} - \hat{a}_{lj}\rvert_{(m-\frac{1}{2})}}{\Delta x_j}\right) \right\}^k.
\label{eq:pressure_iec}
\end{aligned}
\end{equation}
Let $L(\phi_l)$ be a finite-volume discretization of the volume fraction advection equation for phase $l$ [Eq. \eqref{eq:volumef}], given as 
\begin{equation}
L(\phi_l) \equiv \phi_{l}\rvert_m^{k+1} - \phi_{l}\rvert_m^{k} = -\Delta t \left( \frac{\xbar{\phi}_l^{(m+\frac{1}{2})} - \xbar{\phi}_l^{(m-\frac{1}{2})}}{\Delta x_j} \right)^k u_0 + \Delta t  \left( \frac{\hat{a}_{lj}\rvert_{(m+\frac{1}{2})} - \hat{a}_{lj}\rvert_{(m-\frac{1}{2})}}{\Delta x_j} \right)^k.
\label{eq:volume_iec}
\end{equation}
Subtracting Eq. (\ref{eq:pressure_iec}) from the equation $\Big(\sum_{l=1}^2 {L(\phi_l)}/{\alpha_l}\Big)p_0 - \Big(\sum_{l=1}^2 {L(\phi_l)\beta_l}/{\alpha_l}\Big)$, results in $p_{i}\rvert_m^{k+1}=p_0$, $\forall m$, which concludes the proof.
\end{proof}

{Here, a stiffened gas EOS has been used in Eq. \eqref{eq:pressure_iec}; however, a more general cubic EOS can also be used to prove IEC.}
Note that when Eq. (\ref{eq:pressure_iec}) was subtracted from $\Big(\sum_{l=1}^2 {L(\phi_l)}/{\alpha_l}\Big)p_0 - \Big(\sum_{l=1}^2 {L(\phi_l)\beta_l}/{\alpha_l}\Big)$, it was assumed that the split-numerical flux forms used in $\hat{I}_j$ and $\hat{\Phi}_j$ were identical. Without this, it is not possible to show that $p_{i}\rvert_m^{k+1}=p_0$, and therefore, the IEC cannot be proved. Section \ref{sec:ke_consistency} describes that the split-numerical flux forms used in $\hat{\Phi}_j$ and $\hat{C}_j$ should be identical. Therefore, from the transitive property, the split-numerical flux forms used in $\hat{I}_j$ and $\hat{C}_j$ should be identical for the IEC to be satisfied. Similarly, the split-numerical flux forms used in $\hat{H}_j$ and $\hat{F}_j$ should be identical for the IEC to be satisfied.

\section*{Appendix C: Effect of flux splitting}

In this section, the effect of skew-symmetric splitting of the numerical fluxes on the non-linear stability of the method is evaluated. Particularly, the effect of splitting of mass flux $\hat{C}_j$ is assessed for two-phase flows. In Section \ref{sec:keep}, a quadratic split form
\begin{equation}
\hat{C}_j\rvert_{(m\pm\frac{1}{2})} = \xbar{\rho}^{(m\pm\frac{1}{2})} \xbar{u}_{j}^{(m\pm\frac{1}{2})},
\label{eq:quadratic_mass}
\end{equation}
was used. An alternative would be to use the divergence form for the mass flux $\hat{C}_j$ as
\begin{equation}
\hat{C}_j\rvert_{(m\pm\frac{1}{2})} = \xbar{\rho u_j}^{(m\pm\frac{1}{2})}.
\label{eq:divergence_mass}
\end{equation}

\begin{table}[]
\centering
\begin{tabular}{@{}|l|l|l|@{}}
\toprule
                                      & divergence $\hat{C}_j$          & quadratic split $\hat{C}_j$                                        \\ \midrule
$\hat{C}_j\rvert_{(m\pm\frac{1}{2})}$ & $\xbar{\rho u_j}^{(m\pm\frac{1}{2})}$ & $\xbar{\rho}^{(m\pm\frac{1}{2})} \xbar{u}_{j}^{(m\pm\frac{1}{2})}$ \\
$\hat{M}_{ij}\rvert_{(m\pm\frac{1}{2})}$ &
  $\xbar{\rho u_j}^{(m\pm\frac{1}{2})} \xbar{u}_i^{(m\pm\frac{1}{2})}$ &
  $\xbar{\rho}^{(m\pm\frac{1}{2})} \xbar{u}_j^{(m\pm\frac{1}{2})} \xbar{u}_i^{(m\pm\frac{1}{2})}$ \\
$\hat{K}_j\rvert_{(m\pm\frac{1}{2})}$ &
  $\xbar{\rho u_j}^{(m\pm\frac{1}{2})}     \frac{u_i\rvert_{(m\pm1)} u_i\rvert_{m}}{2}$ &
  $\xbar{\rho}^{(m\pm\frac{1}{2})}     \xbar{u}_j^{(m\pm\frac{1}{2})} \frac{u_i\rvert_{(m\pm1)} u_i\rvert_{m}}{2}$ \\ \bottomrule
\end{tabular}
\caption{Consistent momentum and kinetic energy fluxes for the divergence and quadratic-split form of the mass flux.}
\label{tab:divergence_quadratic}
\end{table}

Using the consistency conditions in Section \ref{sec:consistency_sum}, the consistent momentum and kinetic energy fluxes for the divergence and quadratic-split forms of $\hat{C}_j$ can be written as shown in Table \ref{tab:divergence_quadratic}. 
To compare the two formulations,
under-resolved simulations of two-phase Taylor-Green vortex at infinite $Re$ presented in Section  \ref{sec:nonunity_tg} are repeated using these formulations, while keeping all other fluxes unchanged. The time evolution of the total kinetic energy and the individual entropy of fluids 1 and 2 are plotted in Figure \ref{fig:splitting} for $M_o=0.05$. The initial densities of the two fluids are $\rho_{1o}=0.1$ and $\rho_{2o}=1$. The simulation that uses the divergence form of the flux for $\hat{C}_j$ is not stable and diverges, albeit $\hat{M}_{ij}$ and $\hat{K}_j$ satisfy the consistency conditions; whereas the simulation with the proposed quadratic-split form of the flux for $\hat{C}_j$ is stable due to the reduced aliasing errors associated with these schemes \citep{blaisdell1996effect, chow2003further, kennedy2008reduced}.  

\begin{figure}
    \centering
    \includegraphics[width=\textwidth]{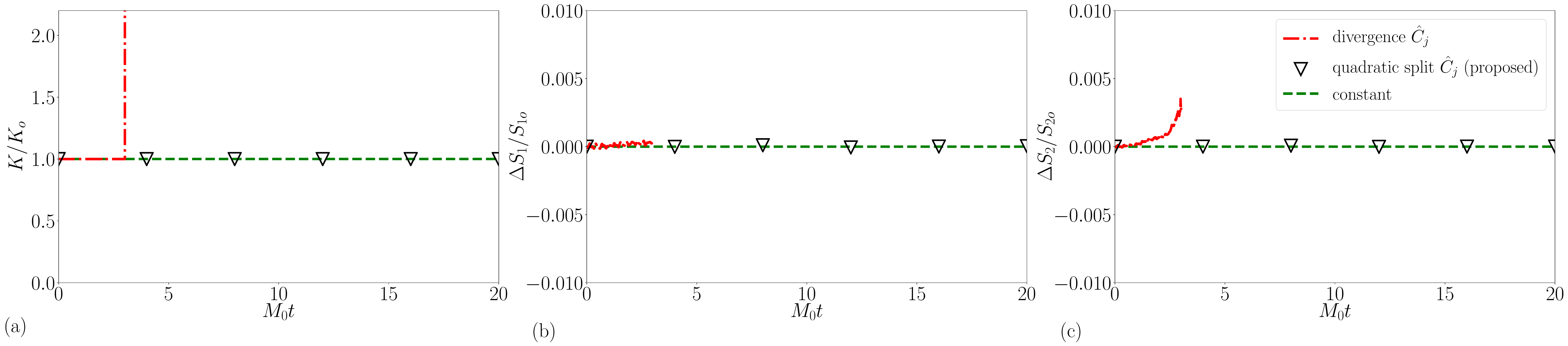}
    \caption{Inviscid two-phase Taylor-Green case with $\rho_2/\rho_1=10$ and $M_o=0.05$. (a) The evolution of the total kinetic energy in the domain $K$ with time normalized by the initial value of the kinetic energy in the domain $K_o$. (b) The change in entropy of fluid $1$ with time, normalized by the initial value. (c) The change in entropy of fluid $2$ with time, normalized by the initial value.}
    \label{fig:splitting}
\end{figure}

\section*{Appendix D: Inconsistent flux formulations}

In Section \ref{sec:under_resolved}, the effect of inconsistent flux formulations for $\hat{I}_j$, that violate the IEC, on the non-linear stability was already examined in detail. Here, a similar study is performed with three different inconsistent flux formulations for $\hat{M}_{ij}$, $\hat{K}_j$, and $\hat{T}_j$, and are compared against the proposed KEEP formulation in Section \ref{sec:keep}. 

An inconsistent flux formulation for $\hat{M}_{ij}$ can be written as 
\begin{equation}
   \hat{M}_{ij}\rvert_{(m\pm\frac{1}{2})} =
  \xbar{\rho u_j}^{(m\pm\frac{1}{2})} \xbar{u}_i^{(m\pm\frac{1}{2})}.
    \label{eq:inconsistent_M}
\end{equation}
This form of $\hat{M}_{ij}$ in Eq. \eqref{eq:inconsistent_T} will violate the consistency condition 1 in Section \ref{sec:consistency_sum}, when all other fluxes are kept unchanged as in Eqs. \eqref{eq:single_phase_keep_scheme}, \eqref{eq:two_phase_keep_scheme}. 
An inconsistent formulation for $\hat{K}_j$ can be written as 
\begin{equation}
    \hat{K}_j\rvert_{(m\pm\frac{1}{2})} = \xbar{\rho}^{(m\pm\frac{1}{2})} \xbar{u}_j^{(m\pm\frac{1}{2})} \frac{u_i\rvert_{(m\pm\frac{1}{2})} u_i\rvert_{(m\pm\frac{1}{2})}}{2}.
    \label{eq:inconsistent_K}
\end{equation}
This form of $\hat{K}_j$ in Eq. \eqref{eq:inconsistent_T} will also violate the consistency condition 1 in Section \ref{sec:consistency_sum}, when all other fluxes are kept unchanged as in Eqs. \eqref{eq:single_phase_keep_scheme}, \eqref{eq:two_phase_keep_scheme}. 
An inconsistent formulation for $\hat{T}_j$ can be written as 
\begin{equation}
        \hat{T}_j\rvert_{(m\pm\frac{1}{2})} = \left\{ \sum_{l=1}^2 \left(\xbar{\rho_l}^{(m\pm\frac{1}{2})} \reallywidehat{a}_{lj} \rvert_{(m\pm\frac{1}{2})} \right) \right\} \frac{u_i\rvert_{(m\pm\frac{1}{2})} u_i\rvert_{(m\pm\frac{1}{2})}}{2}.
    \label{eq:inconsistent_T}
\end{equation}
This form of $\hat{T}_j$ in Eq. \eqref{eq:inconsistent_T} will also violate the consistency condition 1 in Section \ref{sec:consistency_sum}, when all other fluxes are kept unchanged as in Eqs. \eqref{eq:single_phase_keep_scheme}, \eqref{eq:two_phase_keep_scheme}.

To compare the three inconsistent formulations for $\hat{M}_{ij}$, $\hat{K}_j$, and $\hat{T}_j$ in Eqs. \eqref{eq:inconsistent_M}-\eqref{eq:inconsistent_T} with the proposed KEEP formulation in this work,
under-resolved simulations of two-phase Taylor-Green vortex at infinite $Re$ presented in Section  \ref{sec:nonunity_tg} are repeated using these formulations. The time evolution of the total kinetic energy and the individual entropy of fluids 1 and 2 are plotted in Figure \ref{fig:inconsistent_flux} for $M_o=0.05$. The initial densities of the two fluids are $\rho_{1o}=0.1$ and $\rho_{2o}=1$. The simulations that use the inconsistent formulations in Eqs. \eqref{eq:inconsistent_M}-\eqref{eq:inconsistent_T} are not stable because they don't satisfy the consistency conditions, and hence they diverge; whereas the simulation with the proposed KEEP formulation in Section \ref{sec:keep}, that satisfies the consistency conditions, is stable.  

\begin{figure}
    \centering
    \includegraphics[width=\textwidth]{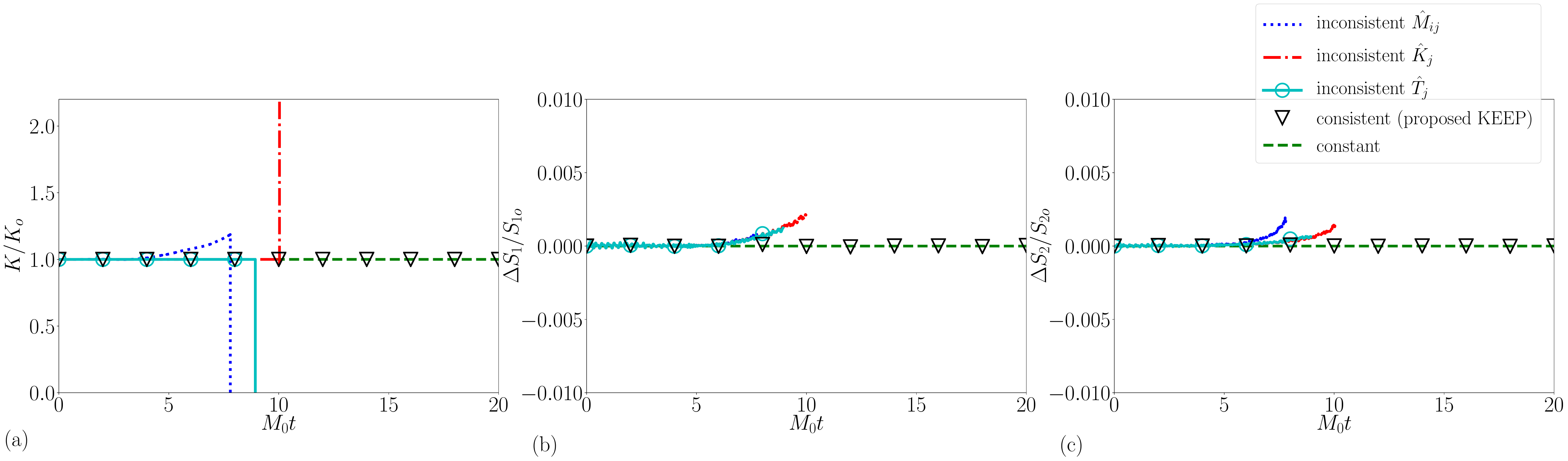}
    \caption{Inviscid two-phase Taylor-Green case with $\rho_2/\rho_1=10$ and $M_o=0.05$. (a) The evolution of the total kinetic energy in the domain $K$ with time normalized by the initial value of the kinetic energy in the domain $K_o$. (b) The change in entropy of fluid $1$ with time, normalized by the initial value. (c) The change in entropy of fluid $2$ with time, normalized by the initial value.}
    \label{fig:inconsistent_flux}
\end{figure}

\bibliographystyle{model1-num-names}
\bibliography{keep.bib}







\end{document}